%% file: eurosys.tex
\documentclass[sigplan]{acmart}

\renewcommand\footnotetextcopyrightpermission[1]{}
\usepackage{url}
\usepackage{amsmath,amsfonts}
\newtheorem{theorem}{Theorem}
\newtheorem{lemma}{Lemma}
\usepackage{subfigure}

\usepackage{graphicx}
\usepackage{textcomp}
\usepackage{xcolor}
\usepackage[justification=centering]{caption}
\usepackage{verbatim}
\usepackage[linesnumbered,ruled,vlined]{algorithm2e}

\AtBeginDocument{%
  \providecommand\BibTeX{{%
    \normalfont B\kern-0.5em{\scshape i\kern-0.25em b}\kern-0.8em\TeX}}}

\setcopyright{acmcopyright}

\begin{document}

\title{CDFGNN: a Systematic Design of Cache-based Distributed Full-Batch Graph Neural Network Training with Communication Reduction}




\author{Shuai Zhang}
\affiliation{%
 \institution{Meituan}
 \city{Beijing}
 \country{China}}
 \email{zhangshuai122@meituan.com}

\author{Zite Jiang}
\affiliation{%
 \institution{SKL Computer Architecture, Institute of Computing Technology, Chinese~Academy~of~Sciences}
 \city{Beijing}
 \country{China}}
 \email{jiangzite19s@ict.ac.cn}
 
 \author{Haihang You*}
\affiliation{%
 \institution{SKL Computer Architecture, Institute of Computing Technology, Chinese~Academy~of~Sciences}
 \city{Beijing}
 \country{China}}
 \email{youhaihang@ict.ac.cn}


\begin{abstract}
Graph neural network training is mainly categorized into mini-batch and full-batch training methods.
The mini-batch training method samples subgraphs from the original graph in each iteration.
This sampling operation introduces extra computation overhead and reduces the training accuracy.
Meanwhile, the full-batch training method calculates the features and corresponding gradients of all vertices in each iteration, and therefore has higher convergence accuracy.
However, in the distributed cluster, frequent remote accesses of vertex features and gradients lead to huge communication overhead, thus restricting the overall training efficiency.

In this paper, we introduce the cached-based distributed full-batch graph neural network training framework (CDFGNN). 
We propose the adaptive cache mechanism to reduce the remote vertex access by caching the historical features and gradients of neighbor vertices.
Besides, we further optimize the communication overhead by quantifying the messages and designing the graph partition algorithm for the hierarchical communication architecture.
Experiments show that the adaptive cache mechanism reduces remote vertex accesses by $63.14\%$ on average.
Combined with communication quantization and hierarchical GP algorithm, CDFGNN outperforms the state-of-the-art distributed full-batch training frameworks by $30.39\%$ in our experiments.
Our results indicate that CDFGNN has great potential in accelerating distributed full-batch GNN training tasks.

\end{abstract}





\maketitle

\input{intro}
\input{background}
\input{cdfgnn}
\input{cache}
\input{quantitative_gp}
\input{exp}
\input{related}
\input{conclusion}

\newpage

\bibliography{eurosys}
\bibliographystyle{IEEEtran}

\end{document}

%% file: intro.tex
\section{Introduction}
With the rise of large-scale pre-training models, the demand for distributed training based on heterogeneous architecture is also increasing. 
As an important deep learning structure, graph neural network (GNN)~\cite{kipf2016semi} has been applied in natural language processing, computer vision, knowledge graphs, etc.
Compared with traditional graph algorithms, the graph neural network often requires computation on heterogeneous devices. Besides, the graph neural network needs to send the features and gradients of vertices across devices in each iteration, which brings huge communication overhead. Therefore, designing an efficient heterogeneous distributed graph neural network training framework is a challenging and engaging research area. 

The training of distributed graph neural network can be categorized into full-batch training~\cite{kipf2016semi,velickovic2018yoshua} and mini-batch training~\cite{hamilton2017inductive,chen2018fastgcn, huang2018adaptive,zeng2019graphsaint,dong2021global}. The main difference between them is whether the entire graph data is involved in each iteration. For the full-batch training method, an iteration contains the model computation phase (including forward propagation and back propagation) and the parameter update phase. For mini-batch training, an additional sampling phase needs to be added. The sampling phase needs to be performed before the model computation phase and the parameter update phase. In the sampling phase, subgraphs are sampled from the entire graph for the current training iteration. Therefore, for the full-batch training, one training epoch is equivalent to one iteration. For the mini-batch training, one training epoch often consists of multiple iterations.

Many mini-batch (sample-based) distributed GNN training methods have been proposed recently. However, these mini-batch training methods lead to problems such as information loss~\cite{cai2021dgcl,jia2020improving,tripathy2020reducing}, additional sampling overhead~\cite{jia2020improving}, and unable to guarantee convergence~\cite{chen2017stochastic}. Therefore, in this paper, we focus on another distributed training strategy: full-batch training.

Compared with traditional graph algorithms or deep learning algorithms, distributed full-batch graph neural network training brings new system-level problems.
The GNN training process has irregular neighbor vertex access and iterative computation at the same time. 
Therefore, graph neural network training is also characterized by both memory access intensive and computing intensive tasks~\cite{thorpe2021dorylus,wang2020gnn}.
In the distributed environment, there is also a problem of intensive communication for the full-batch training methods. 
During the full-batch GNN training, both the model parameters and neighbor vertex data (features and gradients) need to be transmitted across the device. 
Due to the huge communication volume of vertex features and gradients, efficient full-batch GNN training is extremely difficult.

In this paper, we focus on reducing the communication overhead during distributed full-batch graph neural network training. 
Considering that the changes of model parameters during GNN training are usually very slight, we cache historical features and gradients of vertices to reduce the cross-device neighbor vertex access.
In addition, we adopt the quantization method to compress communication messages.
We further design the hierarchical graph partition algorithm to reduce the number of communication messages across physical nodes (at the expense of the extra messages across different GPUs within the same physical node).

Specifically, our main contributions are as follows:
\begin{itemize}
    \item We propose the \textbf{c}ache-based \textbf{d}istributed \textbf{f}ull-batch \textbf{g}raph \textbf{n}eural \textbf{n}etwork training method CDFGNN. By adaptively caching vertex-level historical features and gradients, we can greatly reduce the communication overhead without affecting the convergence accuracy and the number of iterations required for convergence.
    \item We quantify the vertex features and gradients during communication in CDFGNN to further reduce communication overhead.
	\item We design the graph partition algorithm to adapt to the communication characteristics of the hierarchical hardware architecture.
    \item Experiments show that CDFGNN can greatly reduce the communication overhead during distributed full-batch graph neural network training and thus improve the overall training efficiency.
\end{itemize}

This paper is organized as follows:
Section~\ref{sec:back} discusses the challenges of distributed GNN training and explains our motivation.
Section~\ref{sec:CDFGNN} introduces the computation and communication architecture of CDFGNN.
Section~\ref{sec:cache} proposes the adaptive cache mechanism for vertex features and gradients and theoretically proves the convergence of this mechanism. 
Section~\ref{sec:quantify} and section~\ref{sec:hierarchical} describes the quantization method and the hierarchical graph partition algorithm.
Section~\ref{sec:exp} presents and analyzes several experiments, which demonstrate the characteristics and capabilities of CDFGNN.
Finally, we review the related work, conclude our approach, and preview the future project in Section~\ref{sec:related} and Section~\ref{sec:conclusion}.


%% file: background.tex
\section{Background and Motivation}~\label{sec:back}

\subsection{Background}

\begin{figure}[!tb]
	\centering
	\includegraphics[width=0.9\linewidth]{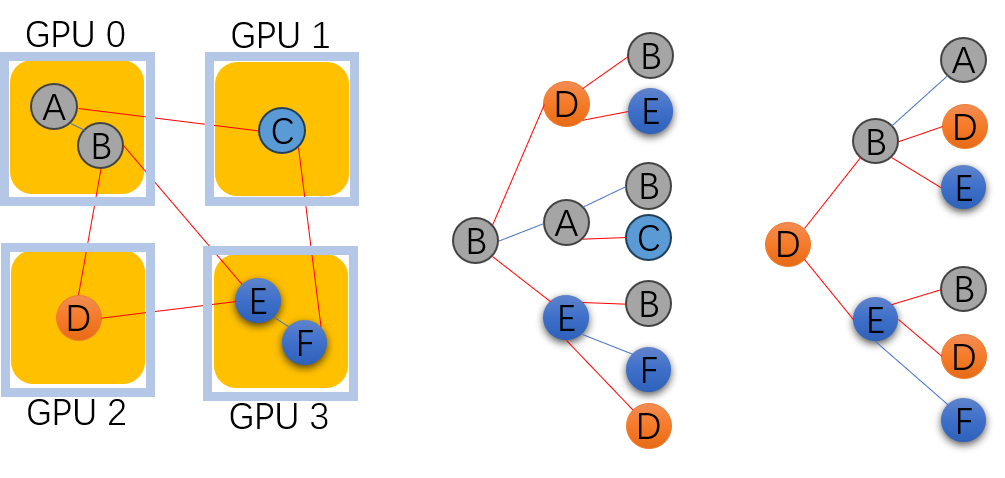}
	\caption{Distributed full-batch GNN Training.}
	\label{fig:full_batch}
\end{figure}


The distributed full-batch GNN training methods require the original graph to be partitioned into several subgraphs, and each computing device (CPU or GPU) only keeps its own subgraph.
The corresponding vertex features are also split and assigned to each device.
Thus, the computation of the entire graph can be completed in just one iteration.

During the training process, each computing device saves a copy of the current model parameters to enable local computation. Therefore, for the full-batch GNN training, the model parameter synchronization is also needed after each iteration.

For both GCN~\cite{kipf2016semi} and GAT~\cite{velivckovic2017graph} models, the vertex features and gradients of all neighbor vertices are required to calculate the features and gradients of the certain vertex during the forward and backward propagation in each layer.
In distributed clusters, such large-scale cross-device data access brings serious communication overhead and becomes a bottleneck of the overall computation. 
Besides, load balancing among the various devices is also important. This is because load imbalance not only results in computational load imbalance, but also communication imbalance.

Figure~\ref{fig:full_batch} shows the training process of a distributed graph neural network with $6$ vertices. Vertices on the same device (GPU) are represented by the same color, and red edges identify edges across GPUs.

The right side is the computational graph of the two-layer graph neural network for vertex ``B'' and vertex ``D''. In order to obtain the final vertex features, $7$ and $6$ cross-device communication messages are required for ``B'' and ``D'' respectively. Each message contains high-dimensional vertex features. When performing backward propagation, the same number of vertex gradients is also required. Therefore, cross-device communication becomes an important bottleneck for efficient training. The overall communication overhead may even account for about $80\%$ of the total training time~\cite{cai2021dgcl,gandhi2021p3,tripathy2020reducing}.

For the distributed mini-batch GNN training, we need to sample graphs before model computation. Thus, an iteration of distributed mini-batch training consists of three stages: sampling, model computation, and model parameter synchronization.


These mini-batches can be sampled by the computing device itself, or sampled by a dedicated sampling device. Each computing device independently executes forward propagation and backward propagation on its corresponding subgraph. 
After the computation stage is completed, these computing devices synchronize and accumulate the gradients to update the model parameters.

\begin{figure}[!htbp]
    \centering
    \begin{minipage}{0.5\linewidth}
    \centerline{\includegraphics[width=1\textwidth]{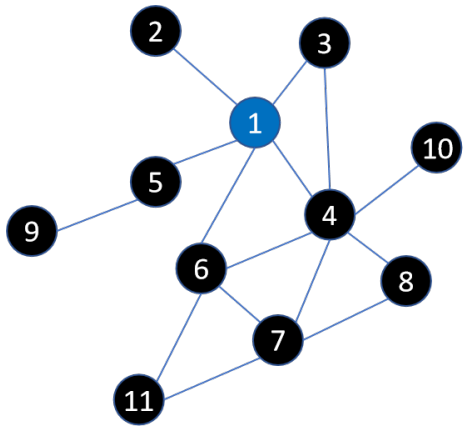}}
	\end{minipage}
	\hfill
	\begin{minipage}{0.35\linewidth}
    \centerline{\includegraphics[width=1\textwidth]{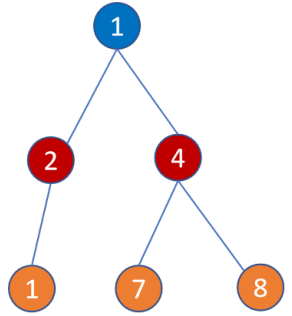}}
	\end{minipage}\\[0.2mm]
	\vfill
    \begin{minipage}{0.5\linewidth}
		\centerline{(a) The original graph}
	\end{minipage}
	\hfill
	\begin{minipage}{0.35\linewidth}
		\centerline{(b) Sampled graph}
	\end{minipage}
    \caption{\quad The sample process of mini-batch training.}
    \label{fig:sample}
\end{figure}

Figure~\ref{fig:sample} shows a $2$-hop sampling process on the original graph.
For the $L$-layers graph neural network, in order to calculate the vertex features, (at least part of) L-hop neighbor vertices need to be included in the sampled subgraph.
In figure~\ref{fig:sample}, for calculating vertex $1$, we additionally add parts of its $2$-hop neighbor vertices to the subgraph.
For graphs with high connectivity and small diameter (such as power-law graphs), even few vertices sampled will generate a large subgraph.
This phenomenon results in significant extra computational overhead.
Although we can restrict the maximum number of sampled neighbor vertices as in figure~\ref{fig:sample}, it will directly reduce the model accuracy.

Compared with the full-batch distributed GNN training, the computation stage of mini-batch training is executed independently on sampled subgraphs, thus avoiding the remote vertex access. 
However, the sampling process also incurs additional computational overhead, including the sampling itself and extra vertex calculations. 
In addition, the mini-batch GNN training often reduces the model accuracy.

\subsection{Motivation}
The frequent and expensive remote neighbor vertex access restricts the scalability of distributed full-batch GNN training.
To overcome this challenge, we can optimize it from the following perspectives:
\begin{itemize}
    \item \textbf{Frequency}: Cache neighbor vertex data instead of executing remote access in each iteration,
    \item \textbf{Expensive}: Compress the message size,
    \item \textbf{Remote}: Make full use of the hierarchical communication architecture.
\end{itemize}

For GNN training tasks, the model parameters tend to stabilize after several training epochs.
Besides, the training process does not require high-precision vertex features and gradients before the model converges.
Therefore, we cache and reuse historical vertex features and gradients during training to reduce communication overhead, especially in the middle stage of the training process.

In order to compress the message size, we quantify the communication messages. These messages include the model parameter gradients and remote neighbor vertex features and gradients. The scale of vertex features and gradients in the GNN training is much larger than the model parameters. 
Meanwhile, when there are small errors in the vertex features and gradients, the final convergence performance will not be significantly reduced, and sometimes it can even prevent the training process from falling into a local optimal solution.
Therefore, we compress the vertex features and gradients during communication by quantifying.

Finally, we analyze the communication characteristics of heterogeneous clusters and find that using the PCIe to communicate between different GPUs in the same physical node is more efficient (higher bandwidth and lower latency) than network communication (InfiniBand) across physical nodes.
Therefore, we propose a graph partition algorithm to reduce the number of messages across physical nodes at the cost of increasing communication within physical nodes.

%% file: cdfgnn.tex
\section{CDFGNN Architecture}~\label{sec:CDFGNN}
In this section, we take the graph convolutional network (GCN) as an example to describe the computation and communication stage of CDFGNN.

\begin{figure}[!tb]
	\centering
	\includegraphics[width=0.9\linewidth]{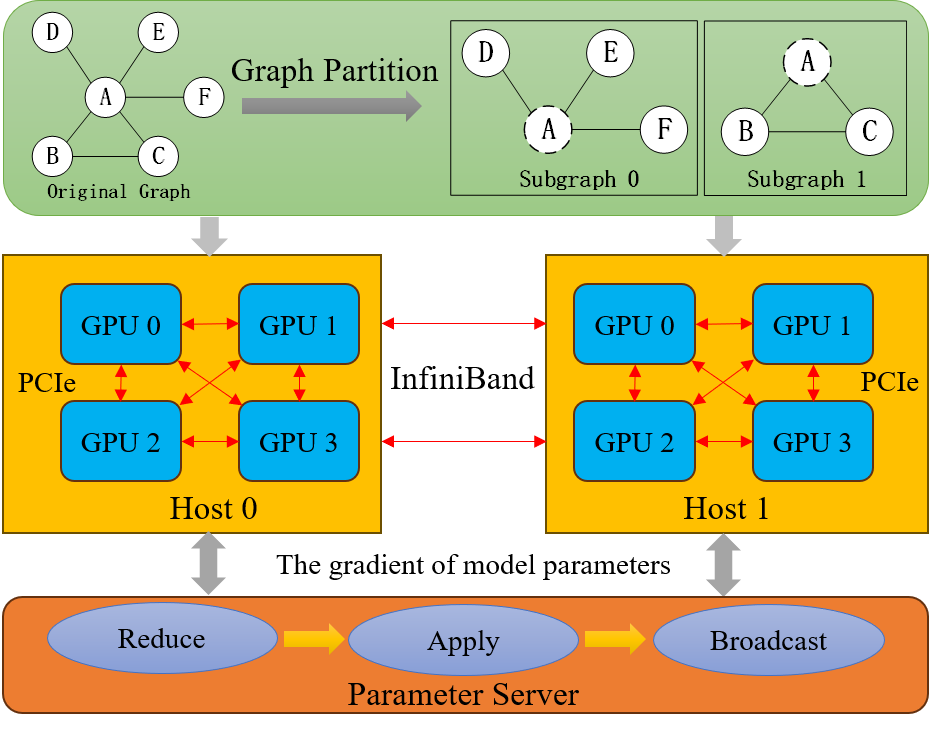}
	\caption{The workflow of CDFGNN.}
	\label{fig5:gnn_comm}
\end{figure}

Figure~\ref{fig5:gnn_comm} shows the overall computing and communication workflow of CDFGNN. CDFGNN first needs to perform the graph partitioning (GP) algorithm to partition the graph (and corresponding input features) into subgraphs equal to the number of computing devices (GPUs).
Different from the traditional full-batch graph neural network training framework, we adopt the \textbf{vertex-cut} GP algorithm. The vertex-cut GP is considered a better approach to handle power-law graphs common in the real world~\cite{gonzalez2012powergraph,chen2019powerlyra}.
Figure~\ref{fig:gather_scatter} demonstrates partition results of the vertex-cut GP algorithm.
In this example, vertex ``B'' exists in all $3$ subgraphs and we choose one of these replicas as the master vertex while others as mirror vertices.

We describe the single iteration distributed training in the algorithm~\ref{alg:CDFGNN}. $L$ refers to the number of layers of the GCN network, and the model parameters of each layer are represented as $W^{(0)}, \cdots, W^{(L-1)}$. Next, we describe the computation and communication stage in detail.

\begin{algorithm}[!tb]
	\caption{CDFGNN Workflow}
	\label{alg:CDFGNN}
    \DontPrintSemicolon
	\KwIn{Graph $G(V,E)$, Sparse Matrix $\hat{A}_i$, Input feature $H^{(0)}_i$, Current Model Parameter $W$.}
	\KwOut{Output Feature $H^{(L)}_i$.}
	
	\For{all process $P(i)$ parallel} {
		$//$ Layer-by-layer forward propagation: \\
		\For{$l = 1, \cdots, L$} {
		    $\ddot{Z}^{(l)}_i \gets \hat{A}_iH^{(l-1)}_iW^{(l-1)}$ \;
		    Synchronize by communication to get $Z^{(l)}_i$ \;
		    $H^{(l)}_i \gets \sigma \left( Z^{(l)}_i \right)$ \;
		}
		
		$//$ Layer-by-layer backward propagation:\\
		Compute Loss Function $\mathcal{L}_i$ and $\ddot{\delta}_i^{(L)}$ \\
		\For{$l = L, \cdots, 1$} {
		    Synchronize by communication to get $\delta_i^{(l)}$. \;
		    $\ddot{\delta}^{(l-1)}_i \gets \delta_i^{(l)} \hat{A}_i \left( W^{(l-1)} \right)^\text{T} \cdot \sigma' \left( Z_i^{(l-1)} \right)$ \;
		    $\nabla_{W^{(l-1)}}\mathcal{L}_i \gets \delta_i^{(l)}\hat{A}_i \left( H_i^{(l-1)} \right)^\text{T}$ \;
		    Parameter Server aggregate $\nabla_{W^{(l-1)}}\mathcal{L}_i$, update and broadcast parameters: \\
		    $W^{(l-1)} = W^{(l-1)} - \eta \sum\limits_{i = 1}^p \nabla_{W^{(l-1)}}\mathcal{L}_i$
		}
	}
\end{algorithm}

\subsection{Computation Stage of CDFGNN}
In the computation stage, each GPU independently performs graph neural network computation tasks on its corresponding subgraph.
We use the BSP model~\cite{valiant1990bridging} to achieve synchronization of vertex features through communication.

Let $A_i$ be the adjacency matrix of subgraph $i$ and $D_i$ be the corresponding submatrix in the original degree matrix.
$\hat{A}_i = D_i^{-1/2}A_iD_i^{-1/2}$ is the normalized adjacency matrix of the subgraph in the computing device $i$.
We use superscript $\ddot{}$ to represent the intermediate matrix values ($\ddot{Z}^{(l)}_i$ and $\ddot{\delta}^{(l-1)}_i$) calculated only from local subgraphs, and the corresponding expressions without this superscript indicate the value ($Z^{(l)}_i$ and $\delta^{(l-1)}_i$) after communication synchronization.

During the forward propagation of GCN, we calculate the vertex feature $H^{(l)}_i$ of the $l$-th layer in the subgraph $i$ as 
\begin{equation}\label{equ:forward_1}
    \ddot{Z}^{(l)}_i = \hat{A}_iH^{(l-1)}_iW^{(l-1)},
\end{equation}
\begin{equation}\label{equ:forward_2}
    H^{(l)}_i = \sigma \left( Z^{(l)}_i \right).
\end{equation}
We calculate $\ddot{Z}^{(l)}_i$ with the local vertex feature $H^{(l-1)}_i$, local normalized adjacency matrix $\hat{A}_i$ and the global model parameter $W^{(l-1)}$.
For restoring the ``real'' $Z_i^{(l)}$ (the same as the value during the sequential training), we need to synchronize and aggregate $\ddot{Z}^{(l)}_i$ from each device through communication.
The communication stage will be introduced in section~\ref{sec:cdfgnn-comm}.

According to $Z^{(l)}_i$, we can calculate the input $H^{(l)}_i$ of the next layer.
$H^{(l)}_i \in \mathbb{R}^{|V_i| \times F_i}$, where $F_i$ refers to the vertex feature dimension of the $i$-th layer. 
By iteratively executing equations~\ref{equ:forward_1} and~\ref{equ:forward_2}, we can complete the calculation of forward propagation layer by layer.

During the backward propagation, we only calculate the loss value of the master vertices when calculating the loss function $\mathcal{L}$.
Thus, we can avoid repeated calculations of gradients on multiple replicas.

We use $\mathcal{L}$ to represent the loss function in the global and $\mathcal{L}_i$ to represent its component on subgraph $i$, while $\sum\limits_{i = 1}^p \mathcal{L}_i = \mathcal{L}$. 
When calculating the gradient, we define $\delta^{(l)}_i = \nabla_{Z^{(l)}_i}\mathcal{L}$ to represent the gradient of the global loss function $\mathcal{L}$ with respect to the global variable $Z^{(l)}_i$, and $\ddot{\delta}^{(l)}_i = \nabla_{\ddot{Z}^{(l)}_i}\mathcal{L}_i$ to represent the gradient of the local loss function $\mathcal{L}_i$ with respect to the local variable $\ddot{Z}^{(l)}_i$. For calculating $\ddot{\delta}^{(l-1)}_i$, we have
\begin{equation}
\begin{aligned}
\ddot{\delta}^{(l-1)}_i = & \frac{\partial \mathcal{L}_i}{\partial \ddot{Z}_i^{(l-1)}} = \frac{\partial \mathcal{L}_i}{\partial Z_i^{(l-1)}} \\
= & \frac{\partial \mathcal{L}_i}{\partial Z_i^{(l)}} \frac{\partial Z_i^{(l)}}{\partial \ddot{Z}_i^{(l)}}
\frac{\partial \ddot{Z}_i^{(l)}}{\partial H^{(l-1)}_i} \frac{\partial H^{(l-1)}_i}{\partial Z_i^{(l-1)}}
\\
= & \delta_i^{(l)} \hat{A}_i \left( W^{(l-1)} \right)^\text{T} \cdot \sigma' \left( Z_i^{(l-1)} \right).
\end{aligned}
\end{equation}
Note that $Z_i^{(l)}$ is calculated with the sum aggregation of $\ddot{Z}_j^{(l)}, j \in [1, p]$, thus we have $\frac{\partial Z_i^{(l)}}{\partial \ddot{Z}_i^{(l)}} = 1$ exists for all subgraph $i$.
Similar with $\ddot{Z}_i^{(l)}$, we can also get $\delta_i^{(l-1)}$ by aggregating $\ddot{\delta}_i^{(l-1)}$ from each device through communication.

With $\delta_i^{(l)}$, the gradient of the model parameter $W^{(l-1)}$ can be calculated as 
\begin{equation} \label{equ5:dw}
    \nabla_{W^{(l-1)}}\mathcal{L}_i = \frac{\partial \mathcal{L}_i}{\partial W^{(l-1)}}
    = \frac{\partial \mathcal{L}_i}{\partial Z_i^{(l)}} \frac{\partial \ddot{Z}_i^{(l)}}{\partial W^{(l-1)}}
    = \delta_i^{(l)}\hat{A}_i \left( H_i^{(l-1)} \right)^\text{T}.
\end{equation}

When performing parameter updates, we need to summarize the gradients calculated on all subgraphs as
\begin{equation}
    W^{(l)} = W^{(l)} - \eta \sum\limits_{i = 1}^p \nabla_{W^{(l)}}\mathcal{L}_i .
\end{equation}
This process also needs to be implemented through communication. 
However, the data size of model parameters is usually much smaller than the data size of neighbor vertex features and gradients.
Thus, the communication overhead of aggregating model parameters is not the performance bottleneck.

In summary, during one iteration (forward + backward) of one GCN layer, there are two communication synchronizations for vertex values (features and gradients). This communication is to obtain the global intermediate value $Z^{(l)}$ in the forward propagation and to obtain the $\delta^{(l)}$ in the backward propagation. Through these communication synchronizations, the calculated model parameter gradients are theoretically consistent with the single-device full-batch training method.


\subsection{Communication Stage of CDFGNN}\label{sec:cdfgnn-comm}
In the real world, most of the data graphs processed by graph neural network algorithms are power-law graphs~\cite{albert2002statistical}, such as social networks, citation graphs, etc.
We adopt the vertex-cut GP algorithm, which is more efficient for power-law graphs.
In figure~\ref{fig:gather_scatter}, we demonstrate the communication pattern for vertex ``B''.
We use the gray vertex in subgraph $1$ to mark this vertex ``B'' as a master vertex, while others are mirror vertices.
In the computation stage, these replicas compute their intermediate values $\ddot{Z}^{(l)}_i$ and $\ddot{\delta}^{(l)}_i$ independently. 
We need to aggregate these values through communication to achieve the same value as when executing on a single device.

\begin{figure}[!tb]
	\centering
	\includegraphics[width=0.9\linewidth]{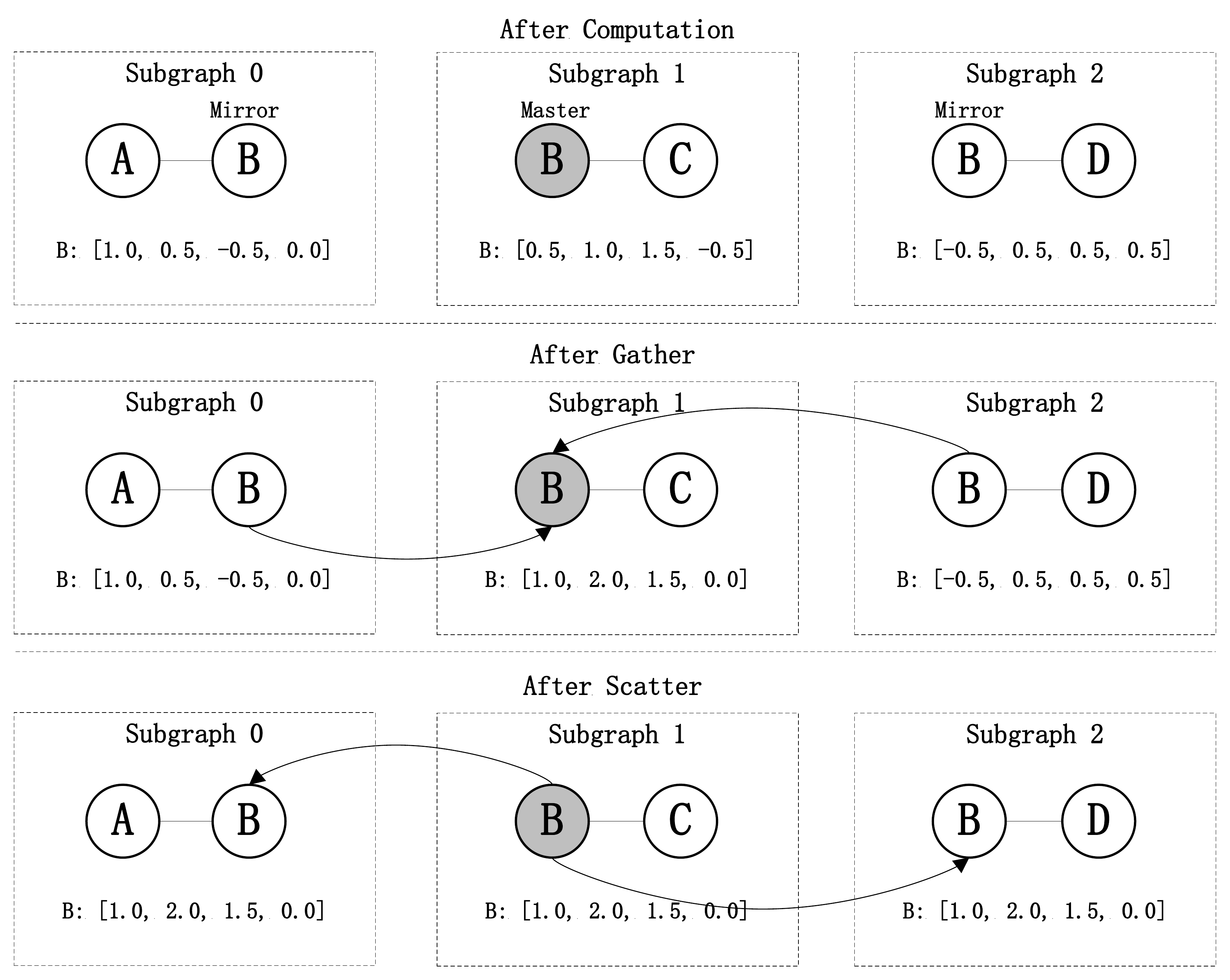}
	\caption{The communication pattern of CDFGNN.}
	\label{fig:gather_scatter}
\end{figure}

CDFGNN takes each vertex as the minimum communication unit.
The communication stage can be divided into two phases: \textbf{gather} and \textbf{scatter}.
In the gather phase, the mirror vertex sends its values to the corresponding master vertex (with the same vertex ID). When the master vertex receives these messages, it should collect 
 them and sum them with its own values.
In the scatter phase, the master vertex sends its aggregated values back to all corresponding mirror vertices.
The mirror vertex uses the received values to replace the original values.
In figure~\ref{fig:gather_scatter}, we list the values of vertex ``B'' at different communication phases in all subgraphs.

This communication pattern requires the mirror vertex to store the location of its master vertex, and the master vertex to store the locations of all its mirrors.
By executing the communication stage, we can ensure that the states of the vertex replicas are consistent with the sequential GNN training.


%% file: cache.tex
\section{Adaptive Vertex Feature Cache}\label{sec:cache}
In this section, we introduce the adaptive cache mechanism of CDFGNN and prove its convergence.

\subsection{Adaptive Cache Mechanism}
In order to reduce the expensive vertex feature and gradient communication during the CDFGNN training process, we propose an adaptive vertex-level caching mechanism.
Specifically, we cache the intermediate variables $Z_i^{(l)}$ and $\delta^{(l)}_i$ during the training process.

For $Z_i^{(l)}$ and $\delta^{(l)}_i$, we adopt the same cache mechanism.
For convenience, we take $Z^{(l)}_i$ as the example to introduce the caching mechanism in detail.
Firstly, we denote $\ddot{Z}^{(l)}_i = \{ z_{i,1}, \cdots, z_{i,|V_i|} \}$, where $z_{i,j}$ represents the feature vector corresponding to the $j$-th vertex of subgraph $i$ in $\ddot{Z}^{(l)}_i$. 
For each subgraph, we renumber the vertices with a local ID for continuous memory access.
The $j$-th vertex here refers to the vertex with local ID $j$ of subgraph $i$.

Let $\tilde{z}_{i,j}$ be the cached value of $z_{i,j}$, and $\tilde{z}_{\cdot, j}$ be the corresponding cached value in $Z^{(l)}_i$.
For each computing device, it should keep the cached value $\tilde{z}_{i,j}$ and $\tilde{z}_{\cdot, j}$ for all vertices in their own subgraphs.


\begin{algorithm}[!tb]
	\caption{Adaptive Vertex Cache Mechanism}
	\label{alg:gnn_cache}
    \DontPrintSemicolon
	\KwIn{current value $z_{i,u}$, cached value $\tilde{z}_{i,u}$ and $\tilde{z}_{\cdot,u}$, threshold $\epsilon$.}
	\KwOut{cached value $\tilde{z}_{i,u}$ and $\tilde{z}_{\cdot,u}$.}
	
	\For{all process $P(i)$ parallel} {
		$//$ Traverse mirror vertices: \\
		\For{$u \in getMirrorVertices()$} {
		    \If{$\|z_{i,u} - \tilde{z}_{i,u}\|_{\infty} > \epsilon \|\tilde{z}_{i,u}\|_{\infty}$} {
                Send the difference value $\Delta_{z_{i,u}} = z_{i,u} - \tilde{z}_{i,u}$ to the corresponding master vertex \\
		        $\tilde{z}_{i,u} \gets z_{i,u}$ \\
		    }
		}
		
		Bulk Synchronize! Wait for messages from all processes to be sent! \\ 
		$//$ Traverse messages and master vertices:\\
		\For{$(u, \Delta_{z_{i,u}}) \in messages$} {
		    $\tilde{z}_{\cdot, u} \gets \tilde{z}_{\cdot, u} + \Delta_{z_{i,u}}$ \\
		    active vertex $u$. \\
		}
		
		\For{$u \in getMaster()$} {
		    \If{$\|z_{i,u} - \tilde{z}_{i,u}\|_{\infty} > \epsilon \|\tilde{z}_{i,u}\|_{\infty}$} {
		        $\tilde{z}_{\cdot, u} \gets \tilde{z}_{\cdot, u} + z_{i,u} - \tilde{z}_{i,u}$  \\
		        $\tilde{z}_{i,u} \gets z_{i,u}$ \\
		        active vertex $u$.
		    }
		}
		
		\For{$u \in$ active vertices} {
            Send the cached value $\tilde{z}_{\cdot, u}$ to the corresponding mirror vertices. \\
		}
	}
\end{algorithm}

The algorithm~\ref{alg:gnn_cache} describes the update strategy of the cached values $\tilde{z}_{i,j}$ and $\tilde{z}_{\cdot, j}$. 
In each forward propagation of the GNN layer, we need to perform this algorithm once.
After the update process is completed, we generate the matrix $Z^{(l)}_i$ by directly combining the cached value $\tilde{z}_{\cdot, j}$.

For the cache mechanism, $\tilde{z}_{i,j}$ keeps the values used by computing device $i$ when building the cached value $\tilde{z}_{\cdot, j}$. When the difference between $\tilde{z}_{i,j}$ and the real value $z_{i,j}$ calculated in current iteration is too large, we need to update $\tilde{z}_{i,j}$ and $\tilde{z}_{\cdot, j}$ for avoiding the large error.
We use $\frac{\|z_{i,u} - \tilde{z}_{i,u}\|_{\infty}}{\|\tilde{z}_{i,u}\| _{\infty}}$ to measure the error. $\|\cdot\|_{\infty}$ is the $L_\infty$ norm, which can be used to represent the maximum absolute value of all elements in it.


We expect $\tilde{z}_{\cdot, j}$ to be consistent across all relevant computing devices.
Thus, when the $\tilde{z}_{i,j}$ of any computing device changed, we need to synchronize it to all other replicas.

In order to increase the proportion of cached values as much as possible without reducing the convergence accuracy or increasing the number of iterations for convergence, we design an adaptive caching mechanism by dynamically adjusting the threshold $\epsilon$.
We update $\epsilon$ by
\begin{equation} \label{equ5:epsilon}
\epsilon  = \left\{ 
    \begin{aligned}
    &\min(\lambda_1 \epsilon, \epsilon + \xi) ,& & acc < mean_{acc} - \mu_1, \epsilon < \nu_1 \cr 
    &\max(\lambda_2 \epsilon, \epsilon - \xi) ,& & acc > mean_{acc} + \mu_2, \epsilon > \nu_2 \cr 
    &\epsilon ,& & otherwise \cr 
    \end{aligned} \right.
\end{equation}
After each iteration, the value of $\epsilon$ is updated. Where $acc$ is the model accuracy on the train set in the current epoch, and $mean_{acc}$ is the exponential moving average of $acc$:
\begin{equation}
mean_{acc} = 0.8 \times mean_{acc} + 0.2 \times acc.
\end{equation}
For the remaining hyperparameters, they are set by default to $\mu_1 = 0.001$, $\mu_2 = 0.02$, $\nu_1 = 0.3$, $\nu_2 = 0.001$, $\xi = 0.01$, $\lambda_1 = 1.05$ and $\lambda_2 = 0.9$ in our experiments.

Among these hyperparameters, we set $\mu_1$ to be much larger than $\mu_2$. This is because in the early stage of training, the accuracy on the training set increases rapidly. 
Only when there is a large enough accuracy increment (larger than $\mu_2$) can we consider that the current cache threshold should be relaxed.
After the model parameters are stabilized, the accuracy of the model on the training set changes slightly.
Therefore, even for small accuracy decreases, the threshold should be set smaller to reduce the cache error. 
In addition, we also use $\xi = 0.02$ to define the maximum step size when $\epsilon$ changes to avoid the error threshold changing too quickly.
We also use $\nu_1$ and $\nu_2$ to limit the value range of $\epsilon$ to $[\nu_2, \nu_1]$.
The settings of these hyperparameters ensure that the training accuracy of the model will not be greatly reduced.

\subsection{Proof of Convergence}
Next, we prove the convergence of the training process when employing the adaptive cache mechanism. 
Specifically, we will prove that after a finite number of iterations, the model parameters $W$ will converge to the local optimal solution $W^*$.
We use the superscript $\tilde{}$ to represent the value obtained in this layer after communication synchronization when the cache mechanism is used.
The values without superscripts represent the values obtained by current model parameters and input features without cache mechanism in all layers.


We first lay out the necessary and basic inequality required for the theoretical analysis.

\begin{lemma}~\label{lemma:basic}
Denote $\| A \|_\infty = \max_{i,j}|A_{i,j}|$, $col(A)$ is the column number of matrix
$A$.
We have $\| A + B \|_\infty \leq \|A\|_\infty + \|B\|_\infty$, $\| A \cdot B \|_\infty \leq \|A\|_\infty \|B\|_\infty$ and $\| A B \|_\infty \leq col(A) \|A \|_\infty \| B\|_\infty$.
\end{lemma}

\begin{proof}
These three inequalities can be proved as follows:
\begin{equation}
\begin{aligned}
    \| A + B \|_\infty &= \max_{i,j}|A_{i,j} + B_{i,j}| \\
           & \leq \max_{i,j}|A_{i,j}| + \max_{i,j}|B_{i,j}| \\
           & = \|A\|_\infty + \|B\|_\infty,
\end{aligned}
\end{equation}
\begin{equation}
\begin{aligned}
    \| A \cdot B \|_\infty &= \max_{i,j}|A_{i,j} \times B_{i,j}| \\
           & \leq \max_{i,j}|A_{i,j}| \times \max_{i,j}|B_{i,j}| \\
           & = \|A\|_\infty \|B\|_\infty,
\end{aligned}
\end{equation}
\begin{equation}
\begin{aligned}
    \| A B \|_\infty &= \max_{i,j} |\sum\limits_{k=1}^{col(A)} A_{i,k} \times B_{k,j} | \\
                     & \leq col(A) \max\limits_{i,j,k}|A_{i,k} \times B_{k,j}| \\
                     & \leq col(A) \max\limits_{i,j}|A_{i,j}| \max\limits_{i,j}|B_{i,j}| \\
                     & = col(A) \|A \|_\infty \| B\|_\infty.
\end{aligned}
\end{equation}
\end{proof}

Next, we state that with bounded staleness on the embeddings, the approximations of the intermediate matrix results are close to the exact ones in the forward propagation.
\begin{lemma}
For the forward propagation of CDFGNN with the cache mechanism, if (a) we have $\| \tilde{Z}^{(l-1)} - Z^{(l-1)} \|_\infty \leq \epsilon_{Z^ {(l-1)}}$, while $\tilde{Z}^{(l-1)}$ and $Z^{(l-1)}$ represent the intermediate values with or without cache mechanism, (b) the function $\sigma(\cdot)$ is $\rho$-Lipschitz continuous, (c) the elements in $Z^{(l)}$, $\hat{A}$ and $W^{(l-1)}$ are bounded, while the absolute values are less than $B$ and the number of columns is less than $C$. Then we have $\|\tilde{H}^{(l-1)} - H^{(l-1)} \|_\infty \leq \rho \epsilon_{Z^{(l-1)} }$ and $\|\tilde{Z}^{(l)} - Z^{(l)} \|_\infty \leq p\nu_1 B + C^2B^2\rho \epsilon_{Z^{ (l-1)}}$.
\end{lemma}

\begin{proof}
We denote $\hat{Z}^{(l)}$ as the intermediate value when the caching mechanism is used in the previous $l-1$ layers, but not used in the $l$-th layer.
Considering that each element in $\hat{Z}^{(l)}$ is the sum from at most $p$ device, the upper bound error of $\tilde{Z}^{(l)}$ for using the cache mechanism in layer $l$, is
\begin{equation}
   \| \tilde{Z}^{(l)} - \hat{Z}^{(l)} \|_\infty \leq p\nu_1 B .
\end{equation}
Where $\nu_1$ is the upper bound of $\epsilon$ defined in the equation (\ref{equ5:epsilon}). 

Therefore, we have
\begin{equation} \label{equ5:lipschitz}
    \begin{aligned}
        \|\tilde{H}^{(l-1)} - H^{(l-1)} \|_\infty = & \| \sigma( \tilde{Z}^{(l-1)}) - \sigma(Z^{(l-1)}) \|_\infty \\
        \leq & \rho \epsilon_{Z^{(l-1)}}
    \end{aligned}
\end{equation}
\begin{equation}
    \begin{aligned}
        \| & \tilde{Z}^{(l)} - Z^{(l)} \|_\infty = \|(\tilde{Z}^{(l)} - \hat{Z}^{(l)}) + 
        (\hat{Z}^{(l)} - Z^{(l)}) \|_\infty \\
        \leq & p\nu_1 B + \|\hat{A} \sigma ( \tilde{Z}^{(l-1)}) W^{(l-1)} - \hat{A}\sigma ( Z^{(l-1)})W^{(l-1)} \|_\infty \\
        \leq & p\nu_1 B + C^2B^2\rho \epsilon_{Z^{(l-1)}}
    \end{aligned}
\end{equation}
The equation (~\ref{equ5:lipschitz}) is obtained from the definition of Lipschitz condition.
\end{proof}

Next, we will prove that the intermediate gradient $\tilde{\delta}^{(l)} = \nabla_{\tilde{Z}^{(l)}} \tilde{\mathcal{L}}$ with cache mechanism is also close to the exact gradient $\delta^{(l)} = \nabla_{Z^{(l)}} \mathcal{L}$.

\begin{lemma} \label{lemma:prove3}
For the backward propagation of CDFGNN, if (a) we have $\| \tilde{Z}^{(l-1)} - Z^{(l-1)} \|_\infty \leq \epsilon_{Z^ {(l-1)}}$, while $\tilde{Z}^{(l-1)}$ and $Z^{(l-1)}$ represent the intermediate values with or without cache mechanism, 
(b) the function $\sigma(\cdot)$ and the derivative of loss function $\nabla \mathcal{L}$ are $\rho$-Lipschitz continuous, (c) the elements in $\delta^{(l)}$, $\hat{A}$, $\sigma'(Z^{(l)})$ and $W^{(l-1)}$ are bounded, and their absolute values are less than $B$ and the number of columns is less than $C$. Then we have $\| \nabla_{\tilde{Z}^{(l)}} \tilde{\mathcal{L}} - \nabla_{Z^{(l)}} \mathcal{L} \|_\infty$ and $ \| \nabla_{W^{(l)}}\tilde{\mathcal{L}} - \nabla_{W^{(l-1)}}\mathcal{L} \|_\infty$ are also bounded.
\end{lemma}

\begin{proof}
First, we prove that $\|\tilde{\delta}^{(l)} - \delta^{(l)} \|_\infty$ is bounded based on the previous lemma.

For the last layer $L$, we have
\begin{equation}
    \| \nabla_{\tilde{Z}^{(L)}} \tilde{\mathcal{L}} - \nabla_{Z^{(L)}} \mathcal{L} \|_\infty
    \leq \rho \epsilon_{Z^{(L)}} .
\end{equation}

Next, we use mathematical induction to complete the proof. For $l' > l$, if it satisfies $\| \nabla_{\tilde{Z}^{(l')}} \tilde{\mathcal{L}} - \nabla_{Z^{(l')}} \mathcal{L} \|_\infty \leq K^{(l')}$, then for the $l$-th layer, we have
\begin{equation}
    \begin{aligned}
        & \| \nabla_{\tilde{Z}^{(l)}} \tilde{\mathcal{L}} - \nabla_{Z^{(l)}} \mathcal{L} \|_\infty \\
        =& \| \tilde{\delta}^{(l + 1)} \hat{A} \left( W^{(l)} \right)^\text{T} \cdot \sigma' \left( \tilde{Z}^{(l)} \right) - \delta^{(l + 1)} \hat{A} \left( W^{(l)} \right)^\text{T} \cdot \sigma' \left( Z^{(l)} \right) \| _\infty \\
        \leq & C^2 \{ \|\tilde{\delta}^{(l + 1)} \|_\infty \|\hat{A} \|_\infty \| \left( W^{(l)} \right)^\text{T} \|_\infty \| \sigma' \left( \tilde{Z}^{(l)} \right) - \sigma' \left( Z^{(l)} \right) \|_\infty  \\
        +&  \|\tilde{\delta}^{(l + 1)} - \delta^{(l + 1)} \|_\infty \|\hat{A} \|_\infty \| \left( W^{(l)} \right)^\text{T} \|_\infty \| \sigma' \left( Z^{(l)} \right) \|_\infty \} \\
        \leq & C^2(B^3 \rho \epsilon_{Z^{(l)}} + K^{(l+1)} B^3) = C^2B^3(\rho \epsilon_{Z^{(l)}} + K^{(l+1)})
    \end{aligned}
\end{equation}
Denote $K^{l} = C^2B^3(\rho \epsilon_{Z^{(l)}} + K^{(l+1)})$, then we can find that the assumption holds for the $l$-th layer.
Therefore, we can complete the proof according to mathematical induction.

For $\| \nabla_{W^{(l)}}\tilde{\mathcal{L}} - \nabla_{W^{(l)}}\mathcal{L} \|_\infty$, we can get it according to the equation (\ref{equ5:dw}):
\begin{equation}
    \begin{aligned}
        \| \nabla_{W^{(l)}}&\tilde{\mathcal{L}} - \nabla_{W^{(l)}}\mathcal{L} \|_\infty \\
        = & \| \tilde{\delta}_i^{(l+1)}\hat{A}_i \left( \tilde{H}_i^{(l)} \right)^\text{T} - \delta_i^{(l+1)}\hat{A}_i \left( H_i^{(l)} \right)^\text{T} \|_\infty  \\
        \leq & C^2 \{ \|\tilde{\delta}_i^{(l+1)}\|_\infty \|\hat{A} \|_\infty 
        \| \left( \tilde{H}_i^{(l)} \right)^\text{T} - \left( H_i^{(l)} \right)^\text{T} \|_\infty\\
        + & \|\tilde{\delta}_i^{(l+1)} - \delta_i^{(l+1)} \|_\infty \|\hat{A} \|_\infty 
        \| \left( H_i^{(l)} \right)^\text{T} \|_\infty \} \\
        \leq & C^2 (B^2 \rho \epsilon_{Z^{(l)}} + K^{l+1}B^2) = C^2B^2 (\rho \epsilon_{Z^{(l)}} + K^{(l+1)})
    \end{aligned}
\end{equation}
\end{proof}

Finally, we will prove that CDFGNN can converge to the local optimal solution under the premise that the error is bounded. For the parameter matrix $W$, we use the subscript $i$ to identify that the value is obtained of the $i$-th iteration.

\begin{theorem}
    For the $L$ layer graph neural network training based on the CDFGNN cache mechanism, given the local optimal parameters $W_{(*)}$ and the initial parameters $W_{(1)}$. Assuming that (a) the activation function $\sigma(\cdot)$ and the derivative of loss function $\nabla \mathcal{L}$ are $\rho$-Lipschitz continuous, (b) the matrix $\hat{A}$, $H$ and $W$, and the corresponding gradients on them are bounded, where the maximum absolute value of the element is $B$, (c) the function $\mathcal{L}(W)$ is $\rho$-smooth.
    We can prove that there is a constant $K > 0$ such that for $\forall N > L_\epsilon$, 
    if the GNN is trained based on the cache mechanism $R$ iterations ($R \in [1, N]$ and is sampled from $[1, \dots, N]$ uniformly) and the learning rate $\eta = \min\left(\frac{1}{\rho}, \frac{1}{\sqrt{N}}\right)$, we have
    \begin{equation}
        \mathrm{E}_R \| \nabla_{W_{(R)}} \mathcal{L} \|^2_F
        \leq 2 \frac{\mathcal{L}(W_{(1)}) - \mathcal{L}(W_{(*)}) + \frac{\rho K}{2}}{\sqrt{N}}.
    \end{equation}
\end{theorem}

\begin{proof}
For the convenience, we denote $\Delta_{(i)} = \nabla_{W_{(i)}}\tilde{\mathcal{L}} - \nabla_{W_{(i)}} \mathcal{L}$.
Considering that the model parameter $W$ is updated under the cache mechanism, we have $W_{(i+1)} = W_{(i)} + \eta \nabla_{W_{(i)}} \tilde{\mathcal{L}}$.
According to lemma~\ref{lemma:prove3} and the $\rho$-smooth property of the function $\mathcal{L}(W)$, we have
\begin{equation}
\begin{aligned}
    \mathcal{L}&(W_{(i+1)}) = \mathcal{L}(W_{(i)} + \eta \nabla_{W_{(i)}} \tilde{\mathcal{L}}) \\
    \leq & \mathcal{L}(W_{(i)}) - \eta \langle \nabla_{W_{(i)}} \mathcal{L},\nabla_{W_{(i)}}\tilde{\mathcal{L}} \rangle
    + \frac{\rho}{2} \eta^2 \| \nabla_{W_{(i)}}\tilde{\mathcal{L}} \|_F^2 \\
    = & \mathcal{L}(W_{(i)}) - \eta \langle \nabla_{W_{(i)}} \mathcal{L}, \Delta_{(i)} \rangle 
    -\eta \| \nabla_{W_{(i)}} \mathcal{L} \|_F^2 \\
    + & \frac{\rho}{2} \eta^2 \left( \| \Delta_{(i)} \|^2_F + \| \nabla_{W_{(i)}} \mathcal{L} \|^2_F
    + 2 \langle \Delta_{(i)}, \nabla_{W_{(i)}} \mathcal{L} \rangle \right) \\
    \leq & \mathcal{L}(W_{(i)}) - (\eta - \frac{\rho}{2} \eta ^ 2) \| \nabla_{W_{(i)}} \mathcal{L} \|^2_F + \frac{\rho}{2} \eta^2 \| \Delta_{(i)} \|^2_F .
\end{aligned}
\end{equation}

The scaling in the last step is based on the value of the learning rate $\eta$.
According to lemma~\ref{lemma:prove3}, we have $\| \Delta_{(i)} \|^2_F \leq \| \nabla_{W_(i)}\tilde{\mathcal{L}} \ |_\infty^2 + \| \nabla \mathcal{L}(W_{(i)}) \|_\infty^2 \leq 2B^2 \leq K$. Therefore, we have
\begin{equation} \label{equ5:13}
    \mathcal{L}(W_{(i+1)}) \leq \mathcal{L}(W_{(i)}) - (\eta - \frac{\rho}{2} \eta ^ 2) \| \nabla_{W_{(i)}} \mathcal{L} \|^2_F + \frac{\rho}{2} \eta^2 K.
\end{equation}
Sum up the equation (\ref{equ5:13}) for $i$ from $1$ to $N$, we can get
\begin{equation} \label{equ5:14}
    (\eta - \frac{\rho}{2} \eta ^ 2) \sum\limits_{i=1}^N \| \nabla_{W_{(i)}} \mathcal{L} \|^2_F \leq
    \mathcal{L}(W_{(1)}) - \mathcal{L}(W^*) +  \frac{\rho}{2} \eta^2 K N.
\end{equation}

Considering $\eta = \min\left(\frac{1}{\rho}, \frac{1}{\sqrt{N}}\right)$, we divide both side of equation (\ref{equ5:14}) by $N(\eta - \frac{\rho}{2} \eta ^ 2)$, then we have
\begin{equation}
\begin{aligned}
    \mathrm{E}_R \| \nabla_{W_{(R)}} \mathcal{L} \|^2_F = & \frac{1}{N} \sum\limits_{i=1}^N \| \nabla_{W_{(i)}} \mathcal{L} \|^2_F \\
    \leq & 2 \frac{\mathcal{L}(W_{(1)}) - \mathcal{L}(W^*) + \frac{\rho}{2} \eta^2KN}{N \eta (2 - \rho \eta)} \\
    \leq & 2 \frac{\mathcal{L}(W_{(1)}) - \mathcal{L}(W^*)}{N \eta} + \rho \eta K \\
    \leq & 2 \frac{\mathcal{L}(W_{(1)}) - \mathcal{L}(W^*) + \frac{\rho K}{2}}{\sqrt{N}}.
\end{aligned}
\end{equation}
When $N \to \infty$, we can find that the expectation of parameter gradient $\mathrm{E}_R \to 0$. Therefore, we show that convergence of parameters can be achieved in finite iterations.

\end{proof}

%% file: quantitative_gp.tex
\section{Communication Quantization} \label{sec:quantify}
In this section, we propose the communication quantization mechanism of CDFGNN. There are many quantization methods, including linear quantization and logarithmic quantification~\cite{daisuke2016convolutional}, exponential quantification~\cite{li2019additive},
differentiable quantization~\cite{gong2019differentiable,yang2019quantization}, etc.
Considering that when we adopt the adaptive cache mechanism, the message sent is the difference value instead of the original value.
Thus, the message data usually follows an uniform distribution. 
For this reason, we adopt the simplest linear quantization method to quantify the difference of vertex features and gradients. We do not quantify the model parameters when communicating with the parameter server.

Specifically, for the calculated difference $\mathbf{m}$ of features or gradients for the vertex $v_i$, it is represented in the form of a 32-bit floating point format in the GPU memory. In order to quantify it into the $B$-bit unsigned integer format, we need to calculate the maximum element value $\max(\mathbf{m})$ and the minimum element value $\min(\mathbf {m})$ at first. Therefore, we can get the quantified value as
\begin{equation}
    q_i =  \left \lfloor \frac{2^B (m_i - \min(\mathbf{m}))}{\max(\mathbf{m}) - \min(\mathbf{m})} + 0.5  \right\rfloor .
\end{equation}
When sending the message, the original message size is $T*L$, and the quantified message size is $B*L + 2T$ (including the maximum and minimum value). Where $L$ refers to the number of elements in $\mathbf{m}$, and $T$ refers to the number of bits of the original data format.

During the recovery, for the quantization value $q_i$, we can restore it to
\begin{equation}
    \tilde{m}_i = \frac{\max(\mathbf{m}) - \min(\mathbf{m})}{2^B} q_i + \min(\mathbf{m}).
\end{equation}
By the definition,  we have
$\left \lfloor \frac{2^B (m_i - \min(\mathbf{m}))}{\max(\mathbf{m}) - \min(\mathbf{m})} - 0.5  \right\rfloor < q_i \leq \left \lfloor \frac{2^B (m_i - \min(\mathbf{m}))}{\max(\mathbf{m}) - \min(\mathbf{m})} + 0.5  \right\rfloor$
. Therefore, the upper bound of the quantization error is $\frac{\max(\mathbf{m}) - \min(\mathbf{m})}{2^{B+1}}$.

\section{Hierarchical Graph Partition Algorithm} \label{sec:hierarchical}
Considering that in the heterogeneous multi-node multi-GPU environment, the communication overhead within a single node and across physical nodes is different.
We demonstrate the communication architecture in figure~\ref{fig5:gnn_comm}.
The GPU is viewed as the basic computing device.

We propose our vertex-cut graph partition algorithm based on the EBV~\cite{zhang2021efficient} algorithm. To adapt to the hierarchical communication architecture,  we rewrite its evaluation function 
\begin{equation} 
\label{equ:eva3}
\begin{aligned}
Eva_{(u, v)}(i) =& (1-\gamma) (\mathbb{I}(i \notin d\_rep_u) + \mathbb{I}(i \notin d\_rep_v)) \\
+& \gamma( \mathbb{I}(host_i \notin h\_rep_u) + \mathbb{I}(host_i \notin h\_rep_v)) \\
+& \alpha \frac{e_{count}[i]}{|E| / p} + \beta \frac{v_{count}[i]}{|V|/p} .
\end{aligned}
\end{equation}
$d\_rep_u$ and $h\_rep_u$ represent the GPU IDs and host (CPU) IDs that vertex $u$ has been assigned. As long as the vertex $u$ has been assigned to any GPU corresponding to the host, the host ID will be added to $h\_rep_u$.
We use $host_i$ to represent the host ID to which the $i$-th GPU belongs.
Besides, $e_{count}[i]$ and $v_{count}[i]$ mean the number of edges and vertices that have been assigned to subgraph $i$.

When partitioning the graph, we assign it edge by edge. For each edge, we select the GPU ID that minimizes the evaluation function as the subgraph ID this edge assigned.


From equation (\ref{equ:eva3}), we can found that the term $\mathbb{I}(host_i \notin host\_rep_u) + \mathbb{I}(host_i \notin host\_rep_v)$ we design can reduce the number of cut vertices between hosts.
Usually, we set $\gamma \ll 1$.
Therefore, this term is mainly worked to select a more reasonable host when the other terms are close. In our experiment, we set $\gamma$ to $0.1$ by default.

For the other terms, $\mathbb{I}(i \notin d\_rep_u) + \mathbb{I}(i \notin d\_rep_v)$ is related to the replication factor among GPUs, while $\alpha \frac{e_{count}[i]}{|E| / p}$ and $\beta \frac{v_{count}[i]}{|V|/p}$ restrict the edge and vertex imbalance factor respectively.
The replication factor is defined as   
\begin{math}
\frac{\sum_{i=1}^{p} |V_{i}|}{|V|}
\end{math}, that represents the average number of replicas for a vertex.
The edge imbalance factor is defined as  
\begin{math}
\frac{\max_{i=1,...,p} |E_{i}|}{|E|/p}
\end{math} 
, while the vertex imbalance factor is defined as 
\begin{math}
\frac{\max_{i=1,...,p} |V_{i}|}{\sum_{i=1}^{p} |V_{i}|/p}
\end{math}
. Both of them are used to measure the balance of partition results.

%% file: exp.tex
\section{Experiments and Analysis} \label{sec:exp}

\begin{table*}[!htb]
	\caption{Statistics of GNN dataset graphs}
	\label{tab:dataset}
	\centering
	\begin{tabular}{ccccc}
	    \hline
	    Dataset & $|V|$ & $|E|$ & Input Dim & Output Dim \\ \midrule
		Reddit & $232,965$ & $11,606,919$ & $602$  & $41$ \\ 
		ogbn-products & $2,449,029$ & $61,859,140$ & $100$  & $47$ \\ 
		ogbn-papers100M & $111,059,956$ & $1,615,685,872$ & $200$ & $172$ \\ 
		Friendster & $65,608,366$  & $1,806,067,135$  & $64$ & $32$ \\
		\bottomrule
	\end{tabular}
\end{table*}
In this section, we test CDFGNN in a heterogeneous environment with multiple physical nodes and multiple GPUs per node. We compare CDFGNN with the state-of-the-art distributed full-batch graph neural network training frameworks on several datasets. 
In addition, we select some representative graph partition algorithms to analyze the influence of different graph partition algorithms on the distributed full-batch GNN training. Finally, we conduct the ablation study to demonstrate the effectiveness of each component.

\subsection{Experiment Setup and Datasets}
In the experiment, we compare CDFGNN with the state-of-the-art distributed full-batch graph neural network training frameworks SANCUS~\cite{peng2022sancus} and CAGNET~\cite{tripathy2020reducing}. We select four datasets: Reddit~\cite{hamilton2017inductive}, ogbn-products~\cite{chiang2019cluster}, ogbn-papers100M~\cite{wang2020microsoft} and Friendster~\cite{Friend} for comparing their performance.
The statistics of these graphs are listed in table~\ref{tab:dataset}.
The Friendster does not provide input features and output categories. 
We randomly generate these data to test the training efficiency of different frameworks on the large graph.

Our experiment platform is a 2-node cluster, with each node has $8$ Nvidia A800 80G GPU. 
The communication within the physical nodes is based on the 16-channel PCIe 4.0, and the communication across the physical nodes is based on the InfiniBand.
We use the NCCL for communication, and list the communication performance in Table~\ref{tab:comm}.


\begin{table}[!tb]
	\caption{Communication Performance between GPUs}
	\label{tab:comm}
	\centering
	\begin{tabular}{ccc}
	    \hline
	    Environment & Pattern & Bandwidth \\ \midrule
        PCIe & Peer2Peer & 22.70 GB/s \\ 
		InfiniBand & Peer2Peer & 8.27 GB/s \\ 
		PCIe & Broadcast & 19.47 GB/s \\ 
		InfiniBand & Broadcast & 11.98 GB/s \\ 
		\bottomrule
	\end{tabular}
\end{table}


We adopt the simple $2$-layer graph convolutional network as our test model. The dimensions of the input and output features are determined by the datasets, while the dimension of the hidden layer is set to $64$ by default. We adopt the cross-entropy function as the loss function, and the Adam optimizer~\cite{kingma2014adam} to update the model parameters. The initial learning rate is set to $0.01$ by default.

\subsection{Distributed Training Efficiency Comparison}
First, we compare CDFGNN with the current state-of-the-art distributed full-batch GNN training frameworks SANCUS and CAGNET. During the training process, we adopted the same GNN model. Meanwhile, we implement CDFGNN with $2$ famous vertex-cut GP algorithms: HEP~\cite{mayer2021hybrid} and DNE~\cite{hanai2019distributed}.
Thus, we can analysis the influence of different graph partition algorithms on the training efficiency.
We also set the $\gamma$ to $0.1$ and $0.0$ respectively for testing the performance of our hierarchical GP algorithm, and represent them as EBV$_{\gamma=0.1}$ and EBV$_{\gamma=0.0}$.
The EBV$_{\gamma=0.0}$ is equivalent to the original EBV algorithm.


\begin{table*}[!htb]
    \caption{The Statistics of differnet graph partition algorithms}
    \centering
    \begin{tabular}{cccccccc}
        \hline
        Dataset & GP algorithm &  Nodes & GPU per Node & Inner  & Outer  & RF & Edge IF \\
        \hline
        reddit & EBV$_{\gamma = 0.0}$ & 2 & 2 & 104217 & 138583 & 2.9027 & 1.0054 \\
        reddit & EBV$_{\gamma = 0.1}$ & 2 & 2 & 105412 & 117879 & 3.0860 & 1.0022 \\
        reddit & HEP & 2 & 2 & 36662 & 52886 & 1.6084 & 1.2693 \\
        reddit & DNE & 2 & 2 & 65578 & 118788 & 2.1025 & 1.2558 \\
        
        ogbn-products & EBV$_{\gamma = 0.0}$ & 2 & 4 & 695905 & 639459 & 3.1788 & 1.0002 \\
        ogbn-products & EBV$_{\gamma = 0.1}$ & 2 & 4 & 952727 & 481147 & 3.3379 & 1.0008 \\
        ogbn-products & HEP & 2 & 4 & 143711 & 127261 & 1.3304 & 1.2323 \\
        ogbn-products & DNE & 2 & 4 & 367460 & 406408 & 1.9363 & 1.1527 \\
        
        friend & EBV$_{\gamma = 0.0}$ & 2 & 8 & 12395102 & 9988776 & 3.7237 & 1.0002 \\
        friend & EBV$_{\gamma = 0.1}$ & 2 & 8 & 19586785 & 5465465 & 4.0322 & 1.0011 \\
        friend & HEP & 2 & 8 & 5794009  &  4675810 &  1.7048 &  1.776 \\
        friend & DNE & 2 & 8 & 8737134  & 11670478 &  2.3546 &  1.7455 \\
        
        papers100M & EBV$_{\gamma = 0.0}$ & 2 & 8 & 20438528 & 16362561 & 3.6503 & 1.0000 \\
        papers100M & EBV$_{\gamma = 0.1}$ & 2 & 8 & 32241760 & 9924817 & 4.0347 & 1.0001 \\
        papers100M & HEP & 2 & 8 & 5826661 & 3085959 & 1.3144 & 2.0204 \\
        papers100M & DNE & 2 & 8 & 10021186 & 13819120 & 2.1475 & 1.3866 \\
        \hline
    \end{tabular}
    \label{tab:part_comp}
\end{table*}

\begin{figure*}[!htb]
	\centering
	\begin{minipage}{0.48\linewidth}
    	\centerline{\includegraphics[width=1\textwidth]{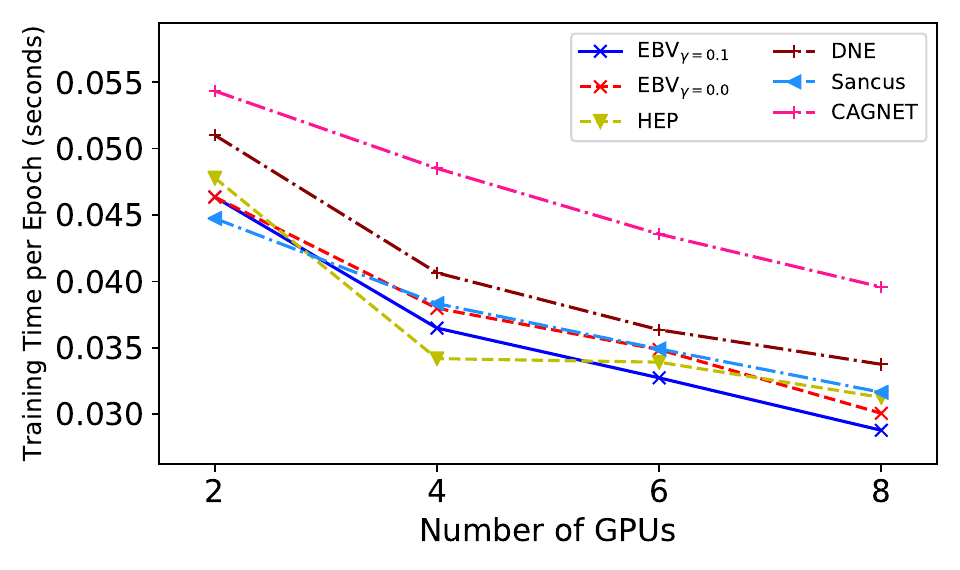}}
	\end{minipage}
	\centering
	\hfill
	\begin{minipage}{0.48\linewidth}
	    \centerline{\includegraphics[width=1\textwidth]{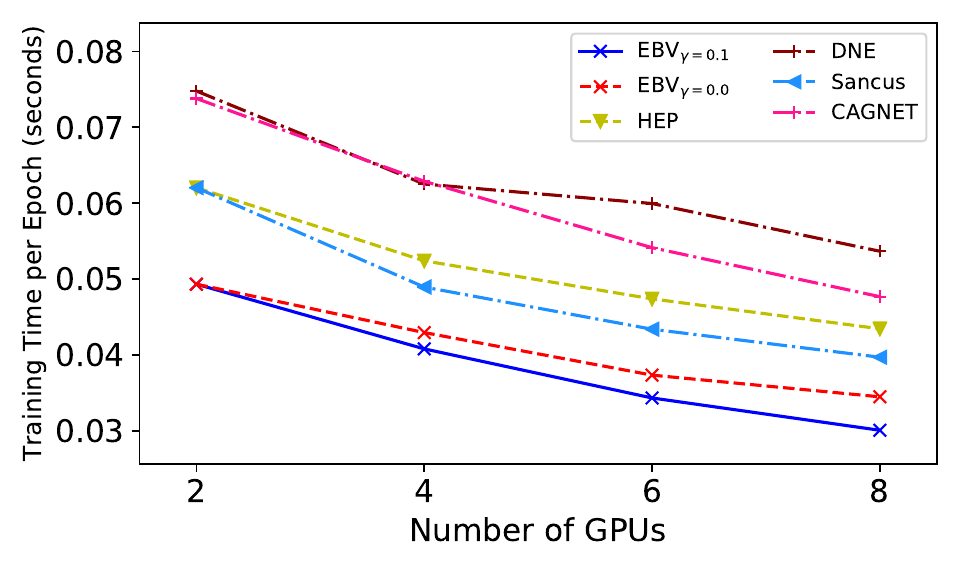}}
	\end{minipage}
	\vfill
	\begin{minipage}{0.48\linewidth}
    	\centerline{(a) Reddit}
	\end{minipage}
	\centering
	\hfill
	\begin{minipage}{0.48\linewidth}
    	\centerline{(b) ogbn-products}
	\end{minipage}
	
    \vfill
	\begin{minipage}{0.48\linewidth}
	    \centerline{\includegraphics[width=1\textwidth]{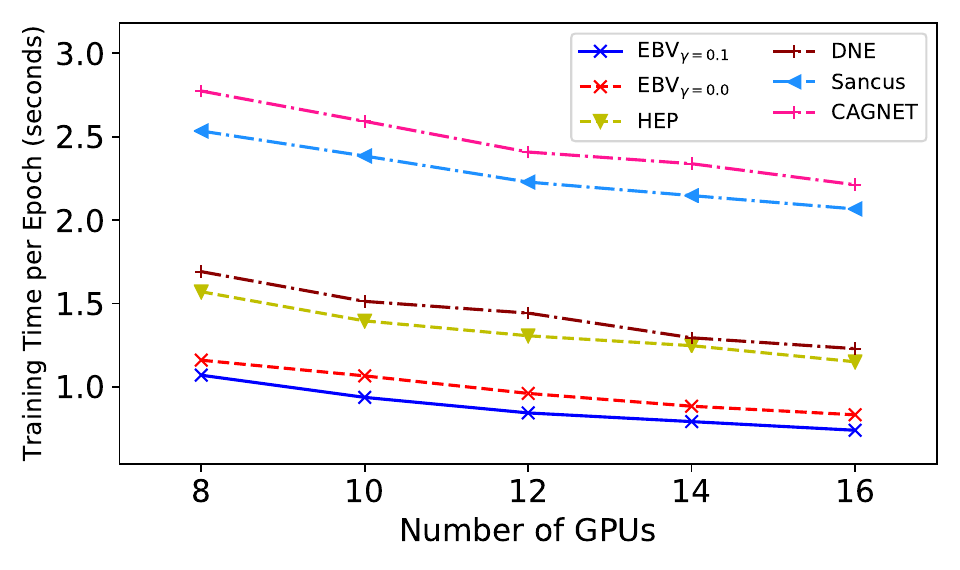}}
	\end{minipage}
    \hfill
	\begin{minipage}{0.48\linewidth}
	    \centerline{\includegraphics[width=1\textwidth]{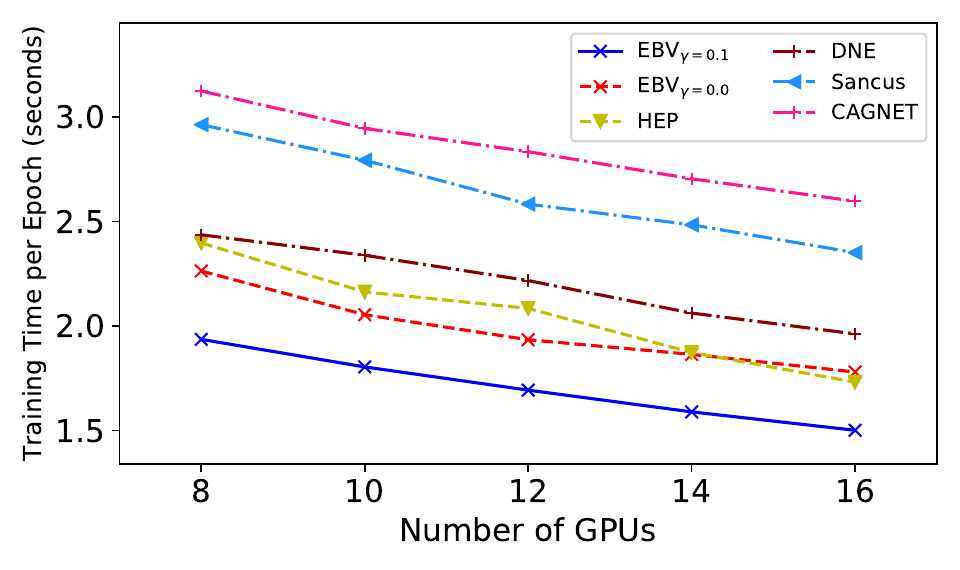}}
	\end{minipage}
	
	\vfill
	\begin{minipage}{0.48\linewidth}
    	\centerline{(c) ogbn-papers100M}
	\end{minipage}
    \hfill
	\begin{minipage}{0.48\linewidth}
    	\centerline{(d) Friendster}
	\end{minipage}
	
	\caption{Comparison of average training time per epoch.}
	\label{fig:timeCompare}
\end{figure*}

Figure~\ref{fig:timeCompare} presents the average training time per epoch for different GNN training frameworks on four datasets. The GPUs we use are evenly distributed on two physical nodes. We use EBV$_{\gamma=0.1}$, EBV$_{\gamma=0.0}$, HEP and DNE to represent the training efficiency when combined with CDFGNN.

From figure~\ref{fig:timeCompare}, we can find that EBV$_{\gamma=0.1}$ achieves the best performance in almost all cases and reduces the training time by $30.39\%$ compared to Sancus on average.
Sancus performs better than CAGNET and even outperforms EBV$_{\gamma=0.1}$ in the smallest case (2 GPUS, Reddit). However, the performance of Sancus is limited for larger cases.
Comparing EBV$_{\gamma=0.1}$ and EBV$_{\gamma=0.0}$, setting $\gamma$ to $0.1$ can achieve better training efficiency on our cluster. 
It is worth noting that when there are only $2$ GPUs, the partition results of EBV$_{\gamma=0.1}$ and EBV$_{\gamma=0.0}$ are equivalent.
The HEP algorithm also performs well in the smallest dataset (Reddit).
But EBV$_{\gamma=0.1}$ leads HEP by a larger margin on other datasets.
Therefore, we believe that the CDFGNN framework combined with the EBV$_{\gamma=0.1}$ can achieve the best training efficiency on graph neural network datasets of different sizes.


\subsection{Ablation Study}
Next, we study the reasons for different performances when different graph partition algorithms are combined with CDFGNN.
We compare graph partition results generated by different GP algorithms in Table~\ref{tab:part_comp}. 
The ``Inner'' and ``Outer'' columns mean the maximum number of inner and outer connections on a single subgraph. The number of inner connections refers to the number of messages within the physical node that need to be sent from this device, and outer connections refer to the messages across the physical nodes.
We also present the replication factor (RF) and edge imbalance factor (Edge IF) defined in Section~\ref{sec:hierarchical} for analyzing. 
Since all GP algorithms compared here are vertex-cut algorithms, we do not give the vertex imbalance factor.

Table~\ref{tab:part_comp} shows the characteristics of all GP algorithms on $4$ datasets.
Setting $\gamma$ to $0.1$ can greatly reduce the number of outer connections ($31.08\%$ on average) at the expense of more inner connections.
Considering the inter and outer communication bandwidth comparison in Table~\ref{tab:comm}, the overall communication overhead can be greatly reduced, thereby improving training efficiency.
The HEP algorithm achieves the smallest inner and outer connections. However, the graph partition results are significantly imbalanced. 
Thus, it leads to imbalanced computing and communication overhead and reduces the overall training efficiency.


We decompose the computation and communication time of different GP algorithms and communication optimization methods based on CDFGNN for further analysis.
We list the computation and communication time per epoch of each GPUs in Figure~\ref{fig:break}. We also provide the corresponding average training time with the dashed lines.
When comparing these GP algorithms EBV$_{\gamma=0.1}$, EBV$_{\gamma=0.0}$, HEP and DNE, all communication optimization methods are used by default.
When comparing the communication optimization methods, the GP algorithm used is EBV$_{\gamma=0.1}$.
The ``Cache'' means only the adaptive cache mechanism is used, while ``Quantify'' means only the communication quantization is used. ``Baseline'' means that no communication optimization methods are used.

\begin{figure*}[!htb]
	\centering
    \begin{minipage}{0.48\linewidth}
    	\centerline{\includegraphics[width=1\textwidth]{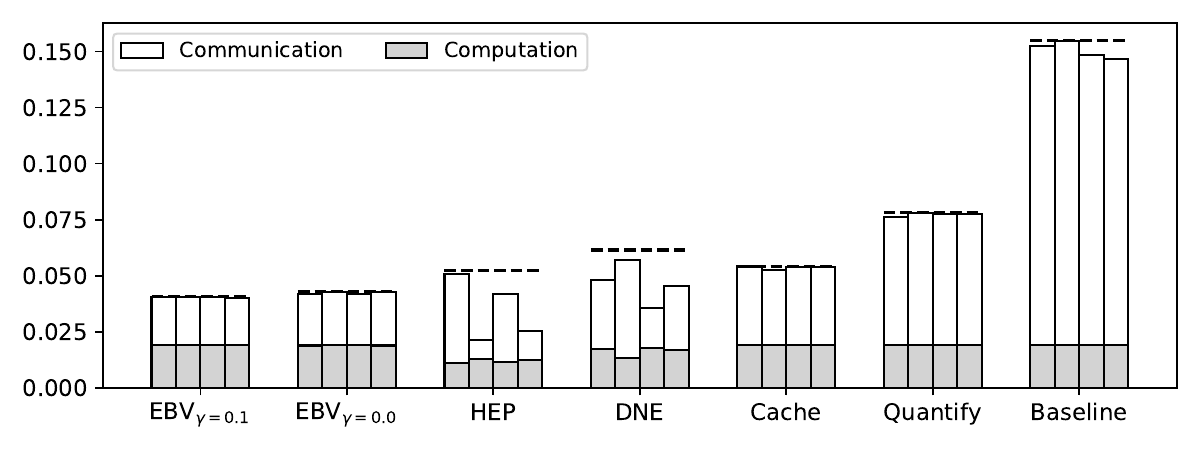}}
	\end{minipage}
	\centering
	\hfill
	\begin{minipage}{0.48\linewidth}
	    \centerline{\includegraphics[width=1\textwidth]{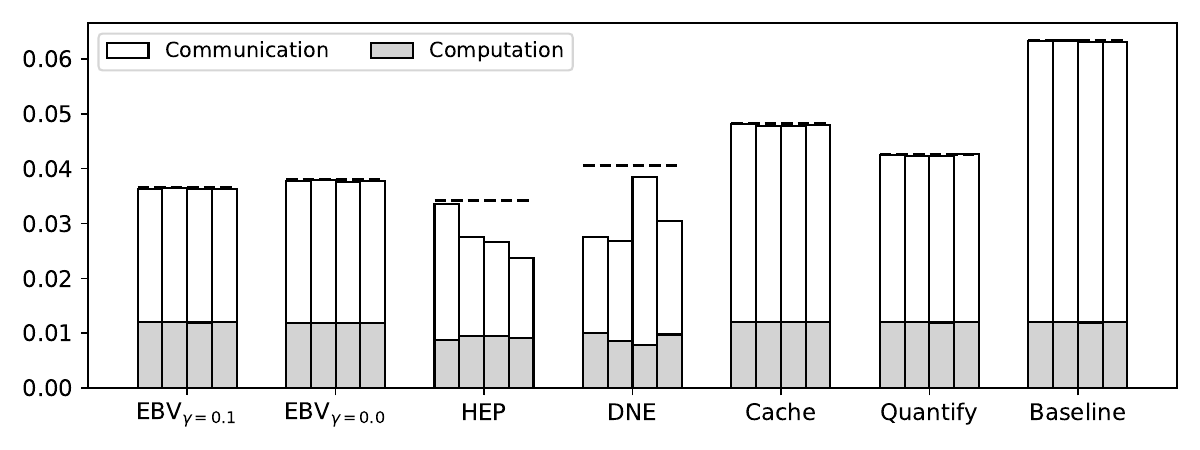}}
	\end{minipage}
	\vfill
	\begin{minipage}{0.48\linewidth}
    	\centerline{(a) ogbn-products}
	\end{minipage}
	\centering
	\hfill
	\begin{minipage}{0.48\linewidth}
    	\centerline{(b) Reddit}
	\end{minipage}
	\caption{Time breakdown of different GP algorithms and communication optimization methods.}
	\label{fig:break}
\end{figure*}

As shown in Figure~\ref{fig:break}, comparing with EBV$_{\gamma=0.1}$, the computation time of EBV$_{\gamma=0.0}$ is roughly the same. However, the communication time of EBV$_{\gamma=0.0}$ is longer. HEP and DNE have significant workload imbalances, thus restricting their training performance.

Meanwhile, both the adaptive cache mechanism and communication quantization can greatly reduce the communication overhead without affecting the computation overhead. 
We include the extra calculation time (quantization and dequantization for communication quantization, caching comparison for adaptive cache mechanism) into the communication time for a fair comparison. Therefore, the communication time is not directly proportional to the number of communication messages.
On ogbn-products, the adaptive cache mechanism achieves better communication optimization, while on Reddit the communication quantification is more efficient.
When combining both methods (EBV$_{\gamma=0.1}$), we achieve the best performance.

We also analysis the message sending percentage of each layer with the adaptive cache mechanism in Figure~\ref{fig:cacheAndsend}. To better understand the cache mechanism during different training epochs, we further provide the cache threshold $\epsilon$.
Figure~\ref{fig:cacheAndsend} shows the sending percentage and cache threshold on ogbn-products and Reddit with $4$ and $8$ GPUs respectively.
It can be found that in the middle stage of training, only few messages are sent, thus greatly reducing communication overhead. This phenomenon is consistent with our hypothesis.
Furthermore, at about $50-100$ training epochs on ogbn-products, almost no vertex features are sent during the forward propagation. Meanwhile, the cache threshold is dynamically adjusted to a larger value in the middle of training and smaller at other times.



\begin{figure*}[!htb]
	\centering

    \begin{minipage}{0.9\linewidth}
		\centerline{\includegraphics[width=1\textwidth]{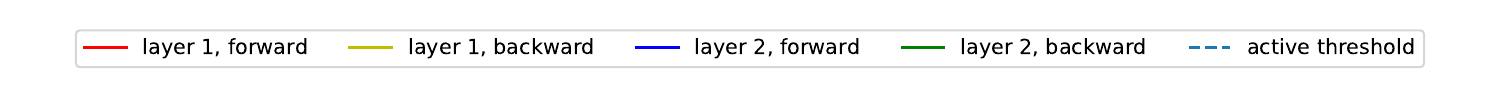}}
	\end{minipage}\\
 
	\begin{minipage}{0.45\linewidth}
    	\centerline{\includegraphics[width=1\textwidth]{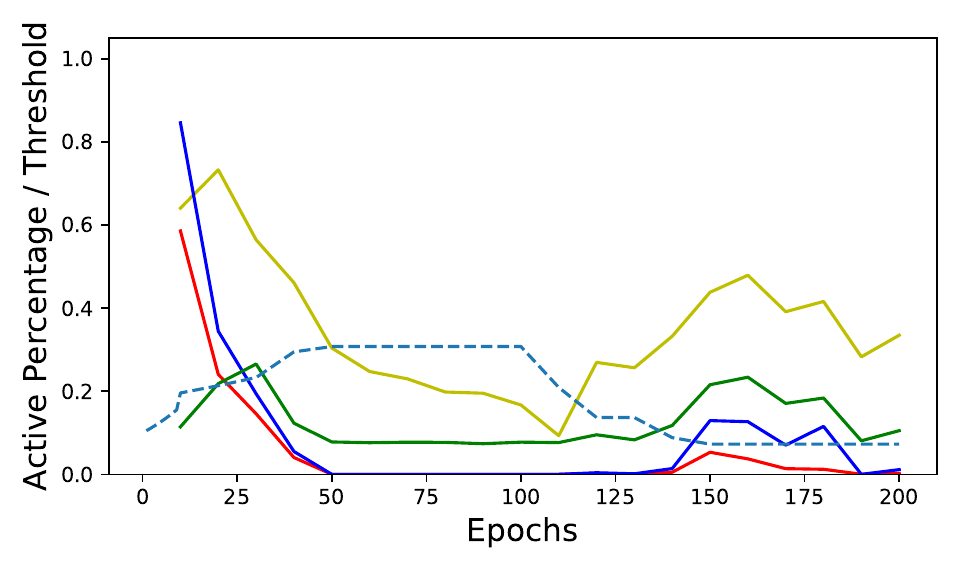}}
	\end{minipage}
	\centering
	\hfill
	\begin{minipage}{0.45\linewidth}
	    \centerline{\includegraphics[width=1\textwidth]{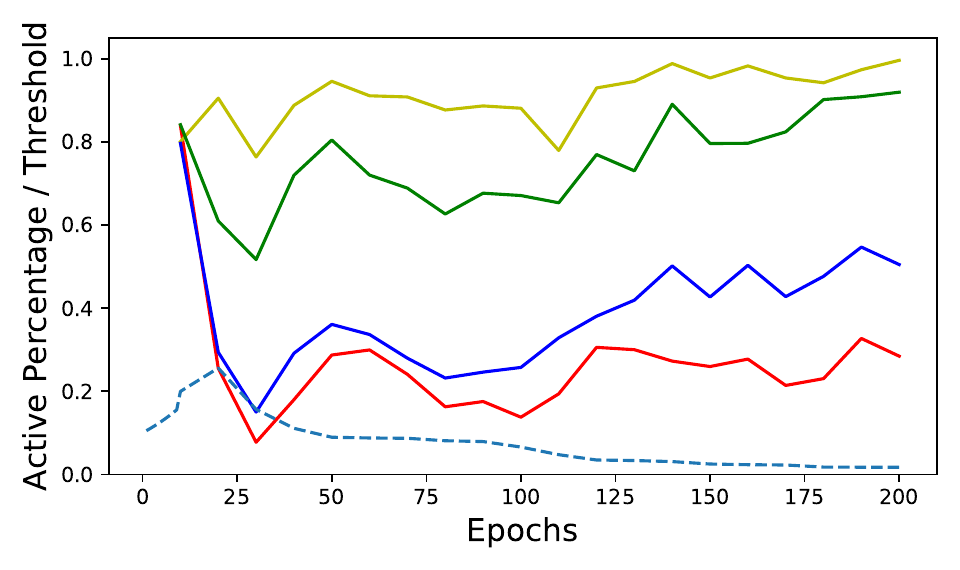}}
	\end{minipage}
	
	\vfill
	\begin{minipage}{0.48\linewidth}
    	\centerline{(a) ogbn-products}
	\end{minipage}
	\centering
	\hfill
	\begin{minipage}{0.48\linewidth}
    	\centerline{(b) Reddit}
	\end{minipage}
	
	\caption{Percentage of cache threshold and sending messages.}
	\label{fig:cacheAndsend}
\end{figure*}

\begin{figure*}[!htb]
    \centering
	\begin{minipage}{0.48\linewidth}
    	\centerline{\includegraphics[width=1\textwidth]{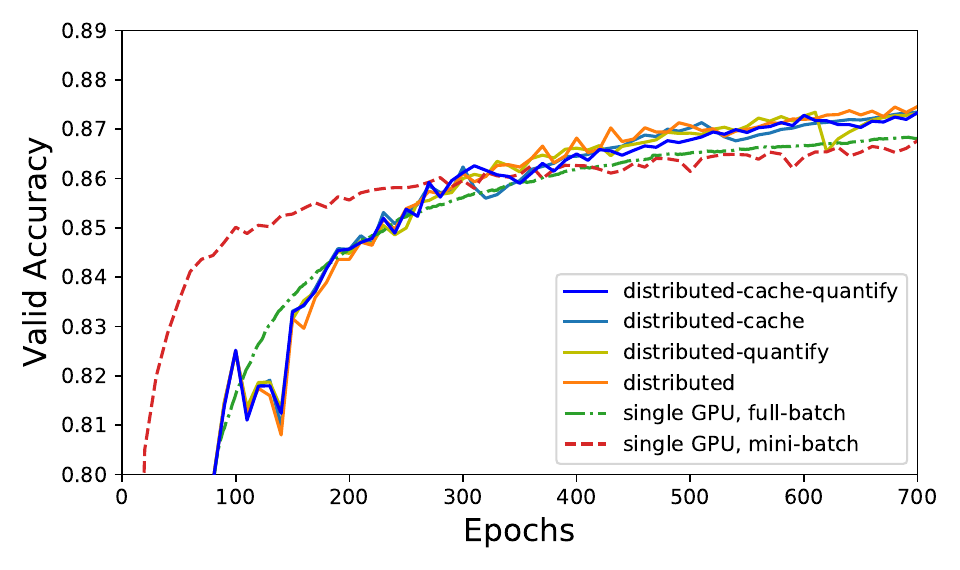}}
	\end{minipage}
	\centering
	\hfill
	\begin{minipage}{0.48\linewidth}
	    \centerline{\includegraphics[width=1\textwidth]{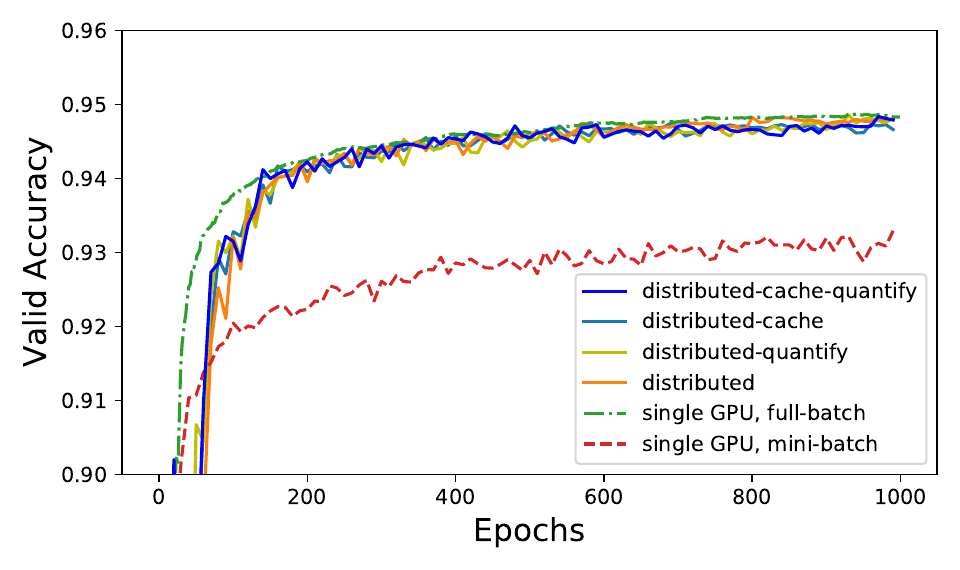}}
	\end{minipage}
	
	\vfill
	\begin{minipage}{0.48\linewidth}
    	\centerline{(a) ogbn-products}
	\end{minipage}
	\centering
	\hfill
	\begin{minipage}{0.48\linewidth}
    	\centerline{(b) Reddit}
	\end{minipage}

	
	
	\caption{The convergence curve of evaluate accuracy.}
	\label{fig:convergence}
\end{figure*}

Finally, we verify the convergence of evaluate accuracy of CDFGNN in Figure~\ref{fig:convergence}. 
In addition to the distributed training approaches of CDFGNN, we also implement the full-batch and mini-batch training methods on the single GPU for comparison. 

The results in figure~\ref{fig:convergence} show that using the adaptive cache mechanism and communication quantification method has almost no impact on the convergence of accuracy.
Due to the small random errors when distributed training, the accuracy in some epochs is even higher than that of single GPU full-batch training. Besides, the mini-batch training method significantly reduce the accuracy, especially on Reddit. That is because we limit the maximum number of neighbors when sampling, and the average degree of Reddit is very large.




%% file: related.tex
\section{Related Work} \label{sec:related}
The research on distributed graph neural network training is still in the early stages~\cite{abadal2021computing}, and only a few these works are based on GPU. 
Compared with traditional distributed large-scale graph computing frameworks~\cite{chen2019powerlyra,fan2017grape,malewicz2010pregel}, the communication overhead of distributed GNN training tasks is more serious. 
This is because the distributed training of each GCN layer or GAT layer requires sending/receiving features and gradients of neighbor vertices, where the dimension of vertex features and gradients is usually very large.

Many existing distributed graph neural network training frameworks adopt the centralized architecture. 
For example, NeuGraph~\cite{ma2019neugraph} proposed a GNN training framework in a single-node multi-GPU environment. They use METIS~\cite{karypis1998fast} as the graph partitioning algorithm, and introduce graph computation optimizations into the management of data partitioning,
scheduling, and parallelism. However, their work is not open source. 
RoC~\cite{jia2020improving} dynamically partitions the graph through an online regression model and proposes a inter-process memory management method, but it also leads to a complex execution workflow. 
PaGraph~\cite{lin2020pagraph} implements static caching of vertices with high degree in GPU memory, and use a special graph partitioning algorithm to balance workload and reduce cross-device data access.
G$^3$~\cite{liu2020g3} utilizes parallel graph optimization to improve graph operations in GPU systems, Grain~\cite{zhang2021grain} selects GNN data by focusing on maximizing social influence, and RDD~\cite {zhang2020reliable} uses unlabeled data. 
AliGraph~\cite{zhu2019aligraph} also uses static caching technology, but only supports CPU clusters.
AGL~\cite{zhang2020agl} uses MapReduce operations to simultaneously optimize the training and inference phases. In order to reduce and balance communication, DistDGL~\cite{zheng2020distdgl} uses a load-balanced graph partitioning algorithm. 
Most of these systems suffer from heavy communication overhead and therefore cannot scale to large-scale applications.
Besides, we should notice that except for NeuGraph and Roc, which support full-batch graph neural network training, other frameworks are mini-batch training methods that require sampling.

%% file: conclusion.tex
\section{Conclusion and Future Work} \label{sec:conclusion}
In this paper, we propose a cache-based distributed full-batch graph neural network training framework CDFGNN. 
To address the problem of excessive communication in existing full-batch training frameworks, we design three optimizations: adaptive cache mechanism, communication quantization, and hierarchical graph partition. 
With these improvements, CDFGNN outperforms the state-of-the-art distributed full-batch training frameworks.
Besides, we theoretically and experimentally prove that the convergence accuracy of CDFGNN is not degraded.
Therefore, we believe that CDFGNN can greatly improve the distributed training efficiency for large-scale graphs.

In the future, we want to make full use of the high-speed communication equipment such as NVLink to further reduce the communication overhead.

%% file: eurosys.bbl
\begin{thebibliography}{10}
\providecommand{\url}[1]{#1}
\csname url@samestyle\endcsname
\providecommand{\newblock}{\relax}
\providecommand{\bibinfo}[2]{#2}
\providecommand{\BIBentrySTDinterwordspacing}{\spaceskip=0pt\relax}
\providecommand{\BIBentryALTinterwordstretchfactor}{4}
\providecommand{\BIBentryALTinterwordspacing}{\spaceskip=\fontdimen2\font plus
\BIBentryALTinterwordstretchfactor\fontdimen3\font minus \fontdimen4\font\relax}
\providecommand{\BIBforeignlanguage}[2]{{%
\expandafter\ifx\csname l@#1\endcsname\relax
\typeout{** WARNING: IEEEtran.bst: No hyphenation pattern has been}%
\typeout{** loaded for the language `#1'. Using the pattern for}%
\typeout{** the default language instead.}%
\else
\language=\csname l@#1\endcsname
\fi
#2}}
\providecommand{\BIBdecl}{\relax}
\BIBdecl

\bibitem{kipf2016semi}
T.~N. Kipf and M.~Welling, ``Semi-supervised classification with graph convolutional networks,'' \emph{arXiv preprint arXiv:1609.02907}, 2016.

\bibitem{velickovic2018yoshua}
P.~Velickovic, G.~Cucurull, A.~Casanova, A.~Romero, and P.~Li{\`o}, ``Yoshua 391 bengio. graph attention networks,'' in \emph{International Conference on Learning Representations}, vol. 392, 2018, p. 393.

\bibitem{hamilton2017inductive}
W.~Hamilton, Z.~Ying, and J.~Leskovec, ``Inductive representation learning on large graphs,'' \emph{Advances in neural information processing systems}, vol.~30, 2017.

\bibitem{chen2018fastgcn}
J.~Chen, T.~Ma, and C.~Xiao, ``Fastgcn: fast learning with graph convolutional networks via importance sampling,'' \emph{arXiv preprint arXiv:1801.10247}, 2018.

\bibitem{huang2018adaptive}
W.~Huang, T.~Zhang, Y.~Rong, and J.~Huang, ``Adaptive sampling towards fast graph representation learning,'' \emph{Advances in neural information processing systems}, vol.~31, 2018.

\bibitem{zeng2019graphsaint}
H.~Zeng, H.~Zhou, A.~Srivastava, R.~Kannan, and V.~Prasanna, ``Graphsaint: Graph sampling based inductive learning method,'' \emph{arXiv preprint arXiv:1907.04931}, 2019.

\bibitem{dong2021global}
J.~Dong, D.~Zheng, L.~F. Yang, and G.~Karypis, ``Global neighbor sampling for mixed cpu-gpu training on giant graphs,'' in \emph{Proceedings of the 27th ACM SIGKDD Conference on Knowledge Discovery \& Data Mining}, 2021, pp. 289--299.

\bibitem{cai2021dgcl}
Z.~Cai, X.~Yan, Y.~Wu, K.~Ma, J.~Cheng, and F.~Yu, ``Dgcl: an efficient communication library for distributed gnn training,'' in \emph{Proceedings of the Sixteenth European Conference on Computer Systems}, 2021, pp. 130--144.

\bibitem{jia2020improving}
Z.~Jia, S.~Lin, M.~Gao, M.~Zaharia, and A.~Aiken, ``Improving the accuracy, scalability, and performance of graph neural networks with roc,'' \emph{Proceedings of Machine Learning and Systems}, vol.~2, pp. 187--198, 2020.

\bibitem{tripathy2020reducing}
A.~Tripathy, K.~Yelick, and A.~Bulu{\c{c}}, ``Reducing communication in graph neural network training,'' in \emph{SC20: International Conference for High Performance Computing, Networking, Storage and Analysis}.\hskip 1em plus 0.5em minus 0.4em\relax IEEE, 2020, pp. 1--14.

\bibitem{chen2017stochastic}
J.~Chen, J.~Zhu, and L.~Song, ``Stochastic training of graph convolutional networks with variance reduction,'' \emph{arXiv preprint arXiv:1710.10568}, 2017.

\bibitem{thorpe2021dorylus}
J.~Thorpe, Y.~Qiao, J.~Eyolfson, S.~Teng, G.~Hu, Z.~Jia, J.~Wei, K.~Vora, R.~Netravali, M.~Kim \emph{et~al.}, ``Dorylus: Affordable, scalable, and accurate $\{$GNN$\}$ training with distributed $\{$CPU$\}$ servers and serverless threads,'' in \emph{15th USENIX Symposium on Operating Systems Design and Implementation (OSDI 21)}, 2021, pp. 495--514.

\bibitem{wang2020gnn}
Z.~Wang, Y.~Guan, G.~Sun, D.~Niu, Y.~Wang, H.~Zheng, and Y.~Han, ``Gnn-pim: A processing-in-memory architecture for graph neural networks,'' in \emph{Advanced Computer Architecture: 13th Conference, ACA 2020, Kunming, China, August 13--15, 2020, Proceedings 13}.\hskip 1em plus 0.5em minus 0.4em\relax Springer, 2020, pp. 73--86.

\bibitem{velivckovic2017graph}
P.~Veli{\v{c}}kovi{\'c}, G.~Cucurull, A.~Casanova, A.~Romero, P.~Lio, and Y.~Bengio, ``Graph attention networks,'' \emph{arXiv preprint arXiv:1710.10903}, 2017.

\bibitem{gandhi2021p3}
S.~Gandhi and A.~P. Iyer, ``P3: Distributed deep graph learning at scale,'' in \emph{15th $\{$USENIX$\}$ Symposium on Operating Systems Design and Implementation ($\{$OSDI$\}$ 21)}, 2021, pp. 551--568.

\bibitem{gonzalez2012powergraph}
J.~E. Gonzalez, Y.~Low, H.~Gu, D.~Bickson, and C.~Guestrin, ``Powergraph: distributed graph-parallel computation on natural graphs.'' in \emph{OSDI}, vol.~12, no.~1, 2012, p.~2.

\bibitem{chen2019powerlyra}
R.~Chen, J.~Shi, Y.~Chen, B.~Zang, H.~Guan, and H.~Chen, ``Powerlyra: Differentiated graph computation and partitioning on skewed graphs,'' \emph{ACM Transactions on Parallel Computing (TOPC)}, vol.~5, no.~3, pp. 1--39, 2019.

\bibitem{valiant1990bridging}
L.~G. Valiant, ``A bridging model for parallel computation,'' \emph{Communications of the ACM}, vol.~33, no.~8, pp. 103--111, 1990.

\bibitem{albert2002statistical}
R.~Albert and A.-L. Barab{\'a}si, ``Statistical mechanics of complex networks,'' \emph{Reviews of modern physics}, vol.~74, no.~1, p.~47, 2002.

\bibitem{daisuke2016convolutional}
M.~Daisuke, H.~L. Edward, and B.~Murmann, ``Convolutional neural networks using logarithmic data representation,'' in \emph{Deep Learning and Unsupervised Feature Learning Workshop, NIPS 2011}, 2016.

\bibitem{li2019additive}
Y.~Li, X.~Dong, and W.~Wang, ``Additive powers-of-two quantization: An efficient non-uniform discretization for neural networks,'' \emph{arXiv preprint arXiv:1909.13144}, 2019.

\bibitem{gong2019differentiable}
R.~Gong, X.~Liu, S.~Jiang, T.~Li, P.~Hu, J.~Lin, F.~Yu, and J.~Yan, ``Differentiable soft quantization: Bridging full-precision and low-bit neural networks,'' in \emph{Proceedings of the IEEE/CVF international conference on computer vision}, 2019, pp. 4852--4861.

\bibitem{yang2019quantization}
J.~Yang, X.~Shen, J.~Xing, X.~Tian, H.~Li, B.~Deng, J.~Huang, and X.-s. Hua, ``Quantization networks,'' in \emph{Proceedings of the IEEE/CVF Conference on Computer Vision and Pattern Recognition}, 2019, pp. 7308--7316.

\bibitem{zhang2021efficient}
S.~Zhang, Z.~Jiang, X.~Hou, Z.~Guan, M.~Yuan, and H.~You, ``An efficient and balanced graph partition algorithm for the subgraph-centric programming model on large-scale power-law graphs,'' in \emph{2021 IEEE 41st International Conference on Distributed Computing Systems (ICDCS)}.\hskip 1em plus 0.5em minus 0.4em\relax IEEE, 2021, pp. 68--78.

\bibitem{peng2022sancus}
J.~Peng, Z.~Chen, Y.~Shao, Y.~Shen, L.~Chen, and J.~Cao, ``Sancus: sta le n ess-aware c omm u nication-avoiding full-graph decentralized training in large-scale graph neural networks,'' \emph{Proceedings of the VLDB Endowment}, vol.~15, no.~9, pp. 1937--1950, 2022.

\bibitem{chiang2019cluster}
W.-L. Chiang, X.~Liu, S.~Si, Y.~Li, S.~Bengio, and C.-J. Hsieh, ``Cluster-gcn: An efficient algorithm for training deep and large graph convolutional networks,'' in \emph{Proceedings of the 25th ACM SIGKDD international conference on knowledge discovery \& data mining}, 2019, pp. 257--266.

\bibitem{wang2020microsoft}
K.~Wang, Z.~Shen, C.~Huang, C.-H. Wu, Y.~Dong, and A.~Kanakia, ``Microsoft academic graph: When experts are not enough,'' \emph{Quantitative Science Studies}, vol.~1, no.~1, pp. 396--413, 2020.

\bibitem{Friend}
``Friendster,'' \url{https://snap.stanford.edu/data/com-Friendster.html}.

\bibitem{kingma2014adam}
D.~P. Kingma and J.~Ba, ``Adam: A method for stochastic optimization,'' \emph{arXiv preprint arXiv:1412.6980}, 2014.

\bibitem{mayer2021hybrid}
R.~Mayer and H.-A. Jacobsen, ``Hybrid edge partitioner: Partitioning large power-law graphs under memory constraints,'' in \emph{Proceedings of the 2021 International Conference on Management of Data}, 2021, pp. 1289--1302.

\bibitem{hanai2019distributed}
M.~Hanai, T.~Suzumura, W.~J. Tan, E.~Liu, G.~Theodoropoulos, and W.~Cai, ``Distributed edge partitioning for trillion-edge graphs,'' \emph{arXiv preprint arXiv:1908.05855}, 2019.

\bibitem{abadal2021computing}
S.~Abadal, A.~Jain, R.~Guirado, J.~L{\'o}pez-Alonso, and E.~Alarc{\'o}n, ``Computing graph neural networks: A survey from algorithms to accelerators,'' \emph{ACM Computing Surveys (CSUR)}, vol.~54, no.~9, pp. 1--38, 2021.

\bibitem{fan2017grape}
W.~Fan, J.~Xu, Y.~Wu, W.~Yu, and J.~Jiang, ``Grape: Parallelizing sequential graph computations,'' \emph{Proceedings of the VLDB Endowment}, vol.~10, no.~12, pp. 1889--1892, 2017.

\bibitem{malewicz2010pregel}
G.~Malewicz, M.~H. Austern, A.~J. Bik, J.~C. Dehnert, I.~Horn, N.~Leiser, and G.~Czajkowski, ``Pregel: a system for large-scale graph processing,'' in \emph{Proceedings of the 2010 ACM SIGMOD International Conference on Management of data}, 2010, pp. 135--146.

\bibitem{ma2019neugraph}
L.~Ma, Z.~Yang, Y.~Miao, J.~Xue, M.~Wu, L.~Zhou, and Y.~Dai, ``$\{$NeuGraph$\}$: Parallel deep neural network computation on large graphs,'' in \emph{2019 USENIX Annual Technical Conference (USENIX ATC 19)}, 2019, pp. 443--458.

\bibitem{karypis1998fast}
G.~Karypis and V.~Kumar, ``A fast and high quality multilevel scheme for partitioning irregular graphs,'' \emph{SIAM Journal on scientific Computing}, vol.~20, no.~1, pp. 359--392, 1998.

\bibitem{lin2020pagraph}
Z.~Lin, C.~Li, Y.~Miao, Y.~Liu, and Y.~Xu, ``Pagraph: Scaling gnn training on large graphs via computation-aware caching,'' in \emph{Proceedings of the 11th ACM Symposium on Cloud Computing}, 2020, pp. 401--415.

\bibitem{liu2020g3}
H.~Liu, S.~Lu, X.~Chen, and B.~He, ``G3: when graph neural networks meet parallel graph processing systems on gpus,'' \emph{Proceedings of the VLDB Endowment}, vol.~13, no.~12, pp. 2813--2816, 2020.

\bibitem{zhang2021grain}
W.~Zhang, Z.~Yang, Y.~Wang, Y.~Shen, Y.~Li, L.~Wang, and B.~Cui, ``Grain: Improving data efficiency of graph neural networks via diversified influence maximization,'' \emph{arXiv preprint arXiv:2108.00219}, 2021.

\bibitem{zhang2020reliable}
W.~Zhang, X.~Miao, Y.~Shao, J.~Jiang, L.~Chen, O.~Ruas, and B.~Cui, ``Reliable data distillation on graph convolutional network,'' in \emph{Proceedings of the 2020 ACM SIGMOD International Conference on Management of Data}, 2020, pp. 1399--1414.

\bibitem{zhu2019aligraph}
R.~Zhu, K.~Zhao, H.~Yang, W.~Lin, C.~Zhou, B.~Ai, Y.~Li, and J.~Zhou, ``Aligraph: A comprehensive graph neural network platform,'' \emph{arXiv preprint arXiv:1902.08730}, 2019.

\bibitem{zhang2020agl}
D.~Zhang, X.~Huang, Z.~Liu, Z.~Hu, X.~Song, Z.~Ge, Z.~Zhang, L.~Wang, J.~Zhou, Y.~Shuang \emph{et~al.}, ``Agl: a scalable system for industrial-purpose graph machine learning,'' \emph{arXiv preprint arXiv:2003.02454}, 2020.

\bibitem{zheng2020distdgl}
D.~Zheng, C.~Ma, M.~Wang, J.~Zhou, Q.~Su, X.~Song, Q.~Gan, Z.~Zhang, and G.~Karypis, ``Distdgl: distributed graph neural network training for billion-scale graphs,'' in \emph{2020 IEEE/ACM 10th Workshop on Irregular Applications: Architectures and Algorithms (IA3)}.\hskip 1em plus 0.5em minus 0.4em\relax IEEE, 2020, pp. 36--44.

\end{thebibliography}
